%% file: cover.tex
\documentclass[]{article}

\usepackage{amsmath}
\usepackage{amsfonts}
\usepackage{amssymb}
\usepackage{xcolor}
\usepackage{amsthm}
\theoremstyle{plain}
\newtheorem{theorem}{Theorem}
\newtheorem{conclusion}[theorem]{Conclusion}
\theoremstyle{definition}
\newtheorem{definition}{Definition}
\newtheorem{notation}{Notation}
\newtheorem{lemma}{Lemma}

\usepackage{graphicx}
\usepackage{enumerate}
\usepackage{wasysym}
\usepackage{lineno}
\usepackage[linesnumbered,ruled]{algorithm2e}
\usepackage{pstricks-add}
\usepackage{tikz,pgfplots}
\pgfplotsset{compat=1.15}
\usepackage{mathrsfs}
\usetikzlibrary{arrows}
\usetikzlibrary[patterns]
\graphicspath{{Figures/}}

\newcommand{\longversion}[1]{}
\newcommand{\shortversion}[1]{#1}

\title{Efficient covering of convex domains by congruent discs}

\author{Shai Gul, Reuven Cohen and Simi Haber}

\begin{document}

\maketitle

\begin{abstract}
In this paper, we consider the problem of covering a plane region with unit discs. We present an improved upper bound and the first nontrivial lower bound on the number of discs needed for such a covering, depending on the area and perimeter of the region. We provide algorithms for efficient covering of convex polygonal regions using unit discs. We show that the computational complexity of the algorithms is pseudo-polynomial in the size of the input and the output. We also show that these algorithms provide a constant factor approximation of the optimal covering of the region. \end{abstract}
%\begin{keyword}
%covering \sep discrete geometry \sep convexity
%\end{keyword}
%\end{frontmatter}
%

\section{Introduction}
In discrete geometry, an efficient covering with unit circles for a given domain is a well known problem, with various applications such as facility location and cellular network design. The question of the optimal covering of a region is also a fundamental question in discrete geometry related to many deep questions on the nature of Euclidean space. In this manuscript we introduce an algorithm which determines the locations of unit discs that induce an efficient covering of a given polygonal domain.
Our approach is based on the properties of the hexagonal regular lattice (the honeycomb), which is the optimal lattice among all lattices in the plane in the density of the covering for a given radius of the enclosing circle (see \cite{Thue}). The presented algorithms are based on placing the centers of the discs at lattice points, where the location and orientation of the lattice relative to the covered region is optimally selected by the algorithms. We describe three algorithms, from the simplest one, achieving good results for convex polygons with low computational complexity, to a more complex one, which guarantees obtaining the optimal lattice-based covering for any convex polygon. \longversion{Subsequently we also present an algorithm achieving the optimal lattice-based covering for any polygon without requiring convexity.}
It should be noted that the provided algorithms are presented for obtaining a covering based on the hexagonal lattice, since it is asymptotically optimal for a fat region, as will be proven below. However, they can be adapted to every given lattice, which may be desirable in some cases. We show that the presented algorithms have polynomial complexity in the combined size of input and output.

Blaschke, T\'oth and Hardwiger \cite{IG} showed that a convex domain with area $A$ and
 perimeter $L$ can be covered by at most
\begin{eqnarray}
  \left\lfloor \frac{2}{3\sqrt{3}}A+\frac{2}{\pi\sqrt{3}}L+1 \right\rfloor\approx\left\lfloor 0.384 \cdot A+0.367 \cdot L+1 \right\rfloor
  \label{eq:1}
 \end{eqnarray}
 unit discs. Their result, using the probabilistic method, is based on estimating the expected number of hexagons in a hexagonal lattice that intersect a domain placed in a random orientation and location (see Section \ref{subsec: IG}). This is a nonconstructive existential result, and in particular it does not provide the desired locations of the discs that yield such a covering. %   To solve this location problem we use an algorithmic approach which leads to the location in a constructive manner. In \cite{DBLP:journals/jocg/BhowmickVX15} deals with placing $n$ equal circles and get the greatest area which being cover. \textcolor{red}{This problem the other direction of}
The algorithms presented in this paper produce the list of discs location that guarantees a covering which achieves this bound. In fact, the second algorithm briefly described guarantees the optimal covering among all coverings based on the hexagonal lattice. We also provide an improved formula for this upper bound on the minimal number of required discs, which strengthens Eq.\ (\ref{eq:1}) above.

Finally, as a consequence of the above, we provide a bound on the approximation ratio between the number of discs in the optimal covering and the number of discs required by the presented algorithms for any region. Taking an asymptotic approach, and defining fat regions as a sequence of regions for which $L = o(A)$, we obtain an asymptotically optimal approximation ratio.  That is, the ratio between the number of discs required by the algorithm to the minimum number of discs required for any covering approaches 1 when the covered region becomes large.

%In the last decade a lots of progression had been done in Network vulnerability due to new methods in computational geometry, where the main idea is to determine the vulnerable points in a given network \cite{PEG,AHKS,Neumayer}. Especially in \cite{Neumayer} has been introduced an algorithm which determines the most vulnerable point in a network by determine finite number of ``special points'' which obtained by intersection between domains. These ``special points'' indictees the desired vulnerable point. In a similar way we will show how to place the lattice on a given domain such that the domain hits a minimal numbers of Voronoi cells each of which is bounded by unit circles and leads to a minimal covering.

\section{Related work}
\cite{DBLP:journals/jocg/GasparTH14,DBLP:journals/jocg/GasparTH14a} described a mechanism for the special case which locates $n$ discs with given radius $r$ to cover a maximum fraction of the area of a unit disc.
The goal of the facility location problem is to locate a minimal number of facilities such that a set of points (or possibly an entire domain) is covered. The first studies of this subject focused on methods of Integral Geometry \cite{Ha,IG,Thue,To1}. With the advances  in computer science, new algorithms and approaches have been developed for the facility location problem, and to the closely related $P$-centers problem.

Megiddo and Supowit  showed that the $P$-centers problem is NP-hard \cite{Meggido1}. In \cite{SUZUKI199669}, a heuristic upper bound to the optimal solution was described in a square. Hwang and Lee \cite{Hwang1993} showed that time complexity of the most efficient algorithm is $O(n^{O(\sqrt{P})})$.
In \cite{Svitkina}, a lower bound for the facility location problem was obtained  and an algorithm achieving a constant approximation ratio was presented.
In \cite{GUHA1999228}, a restricted version of the facility location problem was studied and a constant factor approximation was presented.
 In \cite{ELSHAIKH20161}, a learning mechanism was proposed to solve the $P$-center problem for
a continuous area.

In \cite{DBLP:journals/jocg/BhowmickVX15}, an algorithm for approximating the non-uniform minimum-cost multi-cover problem was described by studying the example of matching clients to servers.

\section{Preliminaries}
\begin{definition}
The \emph{Minkowski sum} of any two sets $A,B \subset \mathbb{R}^2$ is defined to be $A+ B:=\{ x+y:x\in A, y \in B \}$.
For $s \geq 0$, the \emph{Minkowski dilation} by factor $s$ is defined to be
$sA=\{sx:x \in A \}$.
\end{definition}
\begin{figure}[htp]
\centering
\input{Min_Sum}
\caption{The Minkowski sum (the dashed line) of $\varhexagon_0$ and (the reflection of) a convex polygon. The dots demonstrate that a set of points and a translation of another set intersect if and only if the translation resides in the Minkowski sum \cite{IG}.}
\label{fig:min_sim}
\label{img:delta}
\end{figure}
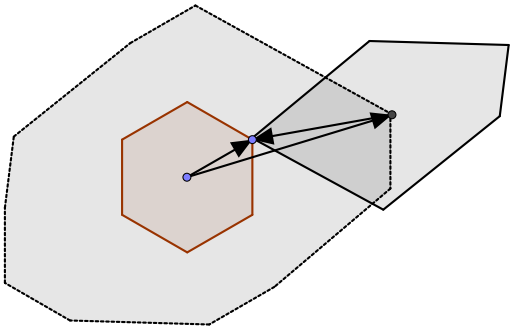
\begin{definition}
For every $\theta \in [0, 2 \pi)$, the \emph{support function} of a domain $\Omega$, denoted $h(\theta)$, is a function that maps every $\theta$ to the supremum over all $p$ such that $L(p, \theta) \cap \Omega \neq \emptyset$, where $L(p, \theta)$ is the line $\{(x,y)|\cos(\theta) x+ \sin (\theta) y =p\}$. The \emph{width} of $\Omega$ is $w(\theta)=h(\theta)+h(\theta+ \pi)$.
\end{definition}
Denote the diameter of a given domain $\Omega$ by $D$. We define  \[\Delta:=[-D-3,D+3]\times [-D-3,D+3] \subset \mathbb{R}^2 . \]
That is, we ensure that $M(\Omega)\subseteq \Delta$ for some rigid motion $M$.

\begin{samepage}
We are interested in the points of the hexagonal regular lattice $\Lambda_h$ contained in $\Delta$:
\begin{notation}
Denote the lattice points in $\Delta$ by $\overline{x}_{mn}:=m \cdot (\sqrt{3},0)+n \cdot (\frac{\sqrt{3}}{2},\frac{3}{2})$ for $m,n \in\mathbb{Z}$\longversion{(see Figure \ref{img:delta})}.
If $m=n=0$, the respective hexagon defined by $\varhexagon_0$. The support function of the hexagon will be denoted by $h_{\varhexagon}(\varphi)$.
\end{notation}
\end{samepage}
We state some standard theorems from integral geometry. We refer the reader to \cite{IG} or other standard textbooks in the field for full proofs.
\begin{theorem}[Cauchy's Formula]
\label{thm:chauchy}
Let $\Omega$ be a bounded convex domain
\begin{align*}
L(\partial \Omega)=\int_0^{2\pi} h(\theta)d \theta=\int_0^\pi w(\theta ) d \theta \ ,
\end{align*}
where  $h$ is the support function of $\Omega$ and $w$ is $\Omega$'s width.
\end{theorem}
\begin{theorem}
Suppose $\Omega$ is a compact, convex domain with a continuous boundary. Then
\begin{equation}
A(\Omega)=\frac{1}{2}\int_0^{2 \pi} ( h^2-h'^2 )d \theta \ ,
\label{eq:1.1}
\end{equation}
where the support's derivative $h'$ may be a distribution.
\end{theorem}
Notice that if $\Omega$ is a polygon, $h$ is not differentiable on a finite set of points and we indeed need to allow $h'$ to be a distribution.

\longversion{
 \begin{figure}[htp]
\centering
\input{../Images/Algorithm1_1.tex}
\caption{The desired grid $\Delta:=[-D-3,D+3]\times [-D-3,D+3]$ }
\label{img:delta}
\end{figure}
}

%\begin{figure}[htp]
%\centering
%\input{../Images/Algorithm1_2.tex}
%\caption{\textcolor{red}{(this image is irrelevant)The desired position of the domain $R(\pi)\cdot\Omega$ for every lattice point in $\Delta$}}
%\label{img:position}
%\end{figure}
\subsection{Upper bounds using integral geometry}
\label{subsec: IG}
The following two theorems are classical results first published by Blaschke. They are brought here for completeness.
\begin{theorem}
Given domains $\Omega_0$ and $\Omega_1$, where $\Omega_1$ is of the form  $\Omega_1=\{(\alpha,\beta,0)\} + \Omega'$, where $(\alpha,\beta)$ is chosen randomly by the uniform distribution in a ball with radius $r$, then $P(\Omega_0 \cap  \Omega_1 \neq \emptyset)= \frac{\mathrm{Area}(\Omega_0 + R(\pi)\cdot \Omega_1)}{\pi r^2}$.
\label{thm:Min_sum}
\end{theorem}
\longversion{
\begin{figure}[htp]
\centering
\input{../Images/Min_su_vec_thesis.tex}
	\caption{Constructing the support function of $K_0 + R(\pi)\cdot K_1$}
\label{fig:support_Min}
\end{figure}
\begin{figure}[htp]
\centering
\input{../Images/con_T(a,b)1_thesis.tex}
	\caption{An intersection between a convex domain $\Omega$ (blue) and a regular hexagon which yields that the center of $\Omega$ in the Minkowski sum of $\left( R(\pi)\cdot \Omega\right) + \varhexagon_{ab}$ and vice versa}
\label{img:construction}
\end{figure}
}
\begin{proof}
(Based on \cite{IG})
Denote by $b_0$ and $b_1$ the centers of $\Omega_0$ and $\Omega_1$, respectively. The vector $z_0$ is the vector which determines $\Omega_0$. In a similar way, we will define for $\Omega_1$.
Given $\Omega_0 \cap \Omega_1 \neq \emptyset$, there exists a point such that $b_0+z_0=b_1+z_1$, i.e., $z_1=z_0+b_0+(-b_1)$ \longversion{(as in Figure~\ref{fig:support_Min}}). Without loss of generality, take $b_0=0$. If $b_1 \in \Omega_1$ then $-b_1 \in R(\pi)\cdot \Omega_1$, so $z_1\in \Omega_0 + R(\pi)\cdot \Omega_1$, which is the area intersection between $\Omega_0$ and $\Omega_1$ in $\mathbb{R}^2$. So the probability in $\mathbb{R}^2$ for the desired intersection is
\[
\frac{\mathrm{Area}(\Omega_0 + R(\pi)\cdot \Omega_1)}{\pi r^2}\;,
\]
 where $r\gg 1$.
\end{proof}

\begin{theorem}
\label{thm:toth}
 Consider $A_1$ as the area of a finite convex region $\Omega_1$ whose perimeter is $L_1$, then the region can be covered using
 \begin{align*}
 \lfloor \frac{2}{3\sqrt{3}}A_1+\frac{2}{\pi\sqrt{3}}L_1+1 \rfloor
 \end{align*}
unit circles, where the square brackets indicate the floor function.
\end{theorem}
\begin{proof}
As a consequence of the above theorem, the mean value of the number of pieces in which a domain $\Omega_1$, limited by a single curve of length $L_1$, is divided when it is put at random on a lattice whose fundamental domains have area $A_0$ and contour of length $L_0$, is
\begin{equation}
\overline{\nu}=\frac{2 \pi(A_0+A_1)+L_0L_1}{2 \pi A_0} \ .
\label{eq:3.3}
\end{equation}
The number $N$ of fundamental domains which have a common point with $\Omega_1$ is always $N \leq \nu$. Consequently, $\overline{N} \leq \nu $ and we get:
Any domain $\Omega_1$ of area $A_1$, limited  by a single curve of length $L_1$, can be covered by a number $\mu$ of fundamental domains of area $A_0$ and contour $L_0$ which satisfies the inequality $\mu \leq \overline{\nu}$, where $\overline{\nu}$ has the value (\ref{eq:3.3}).

If the lattice is that of regular hexagons (which is the optimal lattice) of side $a$, we find that every $\Omega_1$ can be covered by a number of hexagons not exceeding
\begin{equation}
\left[1+\frac{2 L_1}{\sqrt{3} a \pi}+ \frac{2A_1}{2 \sqrt{3}a^2} \right] \ .
\end{equation}
Considering the circles circumscribed about the regular hexagons of side $a$, we obtain the result that every $\Omega_1$ can be covered by this number of circles of radius $a$. Choose $a=1$ and we have the desired result. \end{proof}

\section{Lower bounds for covering}
In this section we prove a lower bound by taking advantage of a few properties of a specific family of multi-graphs.

Let $\Omega$ be a convex region (not necessarily a polygon) with area $A$ and perimeter $L$. We will show the following:
\begin{theorem}
Let $n_\mathrm{opt}$ be the number of unit discs in a minimum covering of $\Omega$. Then
$ n_\mathrm{opt} \geq \max \left\{ \frac{2A}{3\sqrt{3}},\frac{L}{4}\right\}-C$, for some absolute constant $C>0$.
\label{thm:lower_bnd}
\end{theorem}

We prove the area term and the perimeter term separately, beginning with the area. The two parts share a common base insight --- given a covering of the region we consider the Voronoi tessellation defined by taking the centers of the discs of the covering as the seeds of the tessellation.

Let $V_i$ denote the Voronoi cell of the center of the $i$th disc, and let $\Omega^c$ be the complement of $\Omega$.
These regions define a graph $G$ where adjacency in the graph reflects shared edges. We will distinguish between two Voronoi cells types: \emph{inner cells} are the Voronoi cells that do not intersect the boundary of $\Omega$ (except possibly at discrete points) and \emph{exterior cells} are all the cells that have a positive measure intersection with the boundary of $\Omega$. Each of the edges between Voronoi cells is arbitrarily ascribed to only one of the cells, so the cells are disjoint and their union is the entire plane.
Assume there are a total of $n$ cells, of which $m$ are exterior.

We will use the following technical lemma:
\begin{lemma}
Consider a tour along the boundary of the region, starting from an arbitrary point and proceeding counterclockwise. Let $i_1,i_2,i_3,\ldots,i_\ell$ be the sequence of indices of exterior Voronoi cells visited by along the tour. Then $\ell\leq 2m$.
\label{lem:cactus}
\end{lemma}
\begin{proof}
The sequence cannot contain a subsequence of the form $i,\ldots,j,\ldots,i,\ldots,j$ since the line segment between two points on the boundary pertaining to the two visits in the $i$th cell bisects $\Omega$ into two disjoint regions, and thus the line segment between  two points on the boundary pertaining to the two visits in the $j$th cell, who belong to the two different regions, must intersect it. However, by convexity of Voronoi cells, the intersection point must belong to both $V_i$ and $V_j$, that are disjoint, leading to a contradiction.

Thus, the subgraph of $G$ induced by the vertices corresponding to exterior cells is a \emph{cactus graph} --- a connected graph in which any two simple cycles have at most one vertex in common. Indeed, if two simple cycles intersect at two points, $u, v$, then touring along the perimeter we will visit them in the order $u,v,u,v$, or $v,u,v,u$, forming the forbidden pattern.

Cactus graphs can be constructed by starting with a single vertex $i$, and using the following basic steps (here we represent a cactus graph with $m$ vertices by an Euler closed walk of length $t$):
\begin{enumerate}
\item Adding an edge, i.e., replacing some $i$ in the sequence by $i,j,i$, where $j$ does not appear anywhere else in the sequence. This step enlarges $m$ by one and $t$ by two.
\item \label{step: triangle}Adding a triangle, i.e., replacing $i$ in the sequence by $i,j,k,i$, where $j$ and $k$ do not appear anywhere else in the sequence.
    This step enlarges $m$ by two and $t$ by three.
\item Extending an edge, i.e., replacing $i,j$ in the sequence with $i,k,j$ where $k$ does not appear anywhere else in the sequence.
    This step enlarges both $m$ and $t$ by one.
\item Extending a bidirectional edge, that is, replacing $i,j,\ldots,j,i$ in the sequence with $i,k,j,\ldots,j,k,i$ where $k$ does not appear anywhere else in the sequence.
    This step enlarges $m$ by one and $t$ by two.
\end{enumerate}
Thinking of the second step as two mini-steps each incrementing $m$ by one and $t$ by $3/2$, we see that each of these steps increases $m$ by at least one, and none of the steps increases $t$ by more than two. Thus $t \leq 2m$.
\end{proof}

\begin{lemma}
For any covering of $\Omega$ by $n$ unit discs, $n\geq \frac{2(A-C)}{3\sqrt{3}}$, where $C$ is a constant.
\end{lemma}
\begin{proof}
We define the \emph{Voronoi polygon} $U_i$ to be the Voronoin cell $V_i$ if it is internal, and if $V_i$ is external we define it to be the convex polygon (convex hull) formed by the edges between the Voronoi cell $V_i$ and neighboring Voronoi cells and by the points of intersection of $V_i$ with the boundary of $\Omega$. Let $A_i$ be the area of $U_i$.

If $U_i$ is Voronoi polygon with $d_i$ edges, then its area is maximized if it is a regular polygon inscribed in the unit disc, thus $A_i \leq \frac{{d}_i}{2} \sin \frac{2\pi}{{d}_i}$.
By the convexity of the function $x\sin(2\pi/x)$ and and application of Jensen's inequality we get
\[\sum_{i=1}^{n}A_i\leq\sum_{i=1}^{n}\frac{{d}_i}{2} \sin \frac{2\pi}{{d}_i} \leq n  \frac{\overline{d}}{2} \cdot \sin \frac{2\pi}{\overline{d}} \;, \]
where $\overline{d}$ is the average degree in $G$.

Consider the graph of adjacency between the $V_i$s and $\Omega^c$. This is a planar graph with $n+1$ vertices, and thus the sum of degrees is at most $6(n+1)-12=6n-6$. Now, at least $m$ of these edges are incident with $\Omega^c$. Thus the sum of degrees of the Voronoi cells is at most $6n-6-m$. However, each of the exterior cells may have more than one edge shared with $\Omega^c$. The excess degrees obtained by this is at most $t-n\leq n$ (by Lemma \ref{lem:cactus}), which leads to $\sum_{i=1}^m d_i\leq 6m-6-n+n=6m-6$, or $\overline{d}<6$. Thus,
\[\sum_{i=1}^{n}A_i<n\frac{6}{2} \cdot \sin \frac{2\pi}{6}=\frac{3\sqrt{3}}{2}n \;. \]

Now let $C=A - \sum_{i=1}^n A_i$ be the area of $\Omega$ not included in any of the Voronoi polygons. Let $e_i$, $i=1,\ldots,t$ be the exterior edges of the exterior Voronoi polygons.  Let $l_i$ be the lengths of the parts of the boundary of $\Omega$ that are external to their respective $e_i$. by the isoperimetric inequality we have that the maximum area between each $l_i$ and its $e_i$ is obtained when $l_i$ is a circular arc and the area is a circular segment. The area of the segment is given by $C_i=(l_i r_i-r_i^2\sin(l_i/r_i))/2$, where $r_i$ is the radius of curvature. Since the boundary of $\Omega$ is a Jordan curve we have that the total curvature is $\sum_{i=1}^t l_i/r_i=2\pi$. We also have that since each $l_i$ resides in a unit disc, its length is at most $l_i\leq 2\pi$.  Thus, the total area outside the Voronoi polygons is given by:
\[C\leq \sum_{i=1}^t \frac{l_i r_i-r_i^2\sin(l_i/r_i)}{2}\leq \sum_{i=1}^t\frac{l_i^3/r_i}{12}\leq (2\pi)^2\sum_{i=1}^t\frac{l_i/r_i}{12}=\frac{2\pi^3}{3} \;,\]
using the fact that $\sin x\geq x-x^3/6$ for $0\leq x\leq \pi$. Therefore, $C$ is bounded by an absolute constant and the proof is complete.
\end{proof}

For the perimeter term we have the following:
\begin{lemma}
For any covering of $\Omega$ by $n$ unit discs, $n\geq \frac{L}{4}-C$.
\end{lemma}
\begin{proof}
Consider the boundary of $U:=\cup_i U_i$. $U$ is a polygon, which has at most $t\leq 2m\leq 2n$ edges by Lemma \ref{lem:cactus}. For each disc, $i$, having $D_i$ exterior edges we get that the longest total length of the external edges is obtained when the external edges form a regular polygon inscribed in the disc. This polygon has perimeter 2 in the degenerated case $D_i=1$ and perimeter $2D_i\sin(\pi/D_i)$ for $D_i\geq 2$. Assume $k$ cells have $D_i=1$ and the rest of the exterior cells have $D_i\geq 2$. The total length of the external edges is thus at most
\begin{align*}
2k &+\sum_{i=1}^m 2D_i\sin\frac{\pi}{D_i} \leq \\ &\leq 2k+2(m-k)\overline{D}\sin\frac{\pi}{\overline{D}}\leq 2k+2(m-k)\frac{2m-k}{n-k}\sin\frac{\pi(m-k)}{2m-k}\leq 4m  \leq 4n\ ,
\end{align*}
where $\overline{D}$ is the average degree of the cells having at least two exterior edges, and we have applied the Jensen inequality and Lemma \ref{lem:cactus}.

Now each point on the boundary of $\Omega$ has a point at a distance at most two inside $U$, since they fit in the same disc.  Thus, the support function $h_\Omega(\theta)$ of $\Omega$ relative to some point inside $U$ satisfies $h_\Omega(\theta)\leq h_U(\theta)+2$. By Theorem \ref{thm:chauchy}, $L(\partial \Omega)\leq L(\partial U)+4\pi$.
\end{proof}

It should be noted that both the area and perimeter terms are asymptotically sharp, as can be seen in the case of covering a large fat region by discs arranged in a hexagonal lattice configuration (see Section \ref{sec:performance}), and by covering a long narrow (width zero) rectangle by a line of kissing discs, respectively. However, it may be possible to improve Theorem \ref{thm:lower_bnd} by using some combination of the area and perimeter terms which is not the maximum.

\section{Algorithms for covering}
%We would like to find a covering by unit circles. Since every regular hexagon is bounded by a unit circle, the covering by a set of unit circles can be represent by the hexagonal regular lattice which has been proven by T\'oth in the optimal lattice in the plane. We would like to find a special point for locating the domain $\Omega$ leading to a minimal number of regions in the hexagonal regular lattice which intersects the domain $\Omega$. This number of intersections leads to the number of lattice points (in the hexagonal regular lattice) which cover the convex domain. Obviously, the domain $\Omega$ can be placed on the lattice in different positions, every position leads to a different number of intersections. We will show that due to the group of rigid motion, we can shift and rotate the given domain $\Omega$ such that the number of intersections can be minimized, i.e., find the appropriate location.

%Its turn out (see Theorem (\ref{thm:Min_sum})) that the Minkowski sum is a tool for defining a covering for a given lattice which leads to the number of hits of a given domain and the number of fundamental regions in the given lattice.

In this section we aim to find a covering of a convex domain $\Omega$ and non-convex domain by unit discs. Since every hexagon in the hexagonal tiling can be bounded by a unit circle, a covering of $\Omega$ by a subset of the hexagonal lattice gives rise to a covering by unit discs centered at the corresponding lattice points. The hexagonal regular lattice has been proven by T\'oth to be the minimum density lattice to cover the plane and thus it is plausible to use it as the basis for covering. We will show that for fat objects it is indeed asymptotically optimal.

In order to minimize the number of discs, we would like to find the location and orientation for the domain $\Omega$ leading to a minimal number of faces in the hexagonal regular lattice which intersects the domain.
\subsection{Separate orientation and location optimization}
\label{sub:2D}

\shortversion{

\begin{figure}[h!]
\centering
\hspace*{-3cm}
\input{intersect.tex}
\caption{Sets of Minkowski sums (respective to the construction in Fig.\ \ref{fig:min_sim}) in $\Delta$. The algorithm explores the cardinality of all intersection points or corners inside $\varhexagon_0$. The point minimizing the number of Minkowski sums is the desired location for the domain in the hexagonal grid.}
\label{fig:co_nin_sum}
\end{figure}
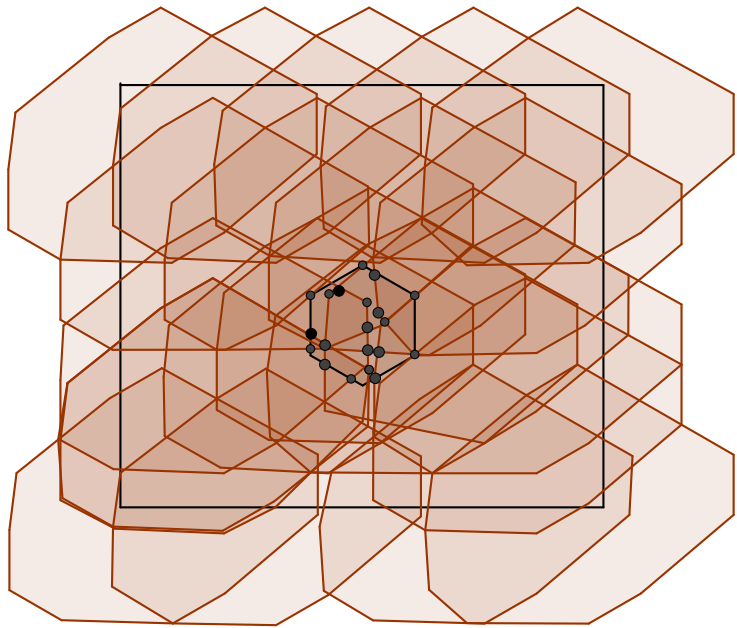
}
\begin{algorithm}
\SetAlgoLined
\KwResult{List of centers of unit discs covering $\Omega$. }
 Given a convex polygon $\Omega$ \newline
 \If{Diameter of minimum enclosing circle of $\Omega\leq 1$ }
 {\Return center of minimum enclosing circle}
\CommentSty{Phase I --- finding $\theta^*$} \newline
  Calculate the width function $w_0$\newline
  Find the diameter of the given domain\newline
  Calculate the Minkowski sum of $(R(\pi+\theta)\cdot\Omega) + \varhexagon_0$ \newline
  Find all the lattice points $\overline{x}_{mn}$ which are contained in $\Delta$ \newline
  Find the optimal angle $\theta$ by minimizing:\
 \[
\theta_0=\arg\min_\theta\left\{w_0(\theta)+w_0\left(\theta+\frac{\pi}{3}\right)+w_0\left(\theta+\frac{2\pi}{3}\right) \quad \left( 0 < \theta \leq  \pi \right) \right\} \
\]\
Place $\Omega$ in orientation with $\theta_0$\\
\CommentSty{Phase II --- finding the centers of covering discs} \newline
$points=\emptyset$ \newline
%For $y_i$ be the \in \partial T_i$ \textcolor{red}{In this case exists 4 points}\\
 \For{$\overline{x}_{mn} \in \Delta$}{
  Find $T_{\overline{x}_{mn}}^\Omega (\theta)$
  }
  \For {$i=1$:$|\overline{x}_{mn}|$}{
    \For {$j=i+1$:$|\overline{x}_{mn}|$}{
        %\If {$y_i=y_j$ \text{and} $T_{i}\neq T_{j}$ \text{and} $y_i \in \varhexagon_0$ }{$points=[points, y_j]$ }}}%
				\For {\textnormal{each intersection point} $y\in\partial T_i\cap\partial T_j\cap \varhexagon_{{}_0}$ }{$points=[points, y_i]$}
        %\If {$\text{Corner point} \in \partial T_{i}^\Omega (\theta)\cap \varhexagon_0$}{$points=[points, y_i]$}
        %\If{$\varhexagon_0\cap \partial T_{i}^\Omega (\theta)\neq \emptyset$}{$points=[points, y_i]$}
				}}
  \For{j=1:$|{points}|$}{
    \For{i=1:$|\overline{x}_{mn}|$}{
        Index(j)=0 \newline
        \If {$points_i \in Int T_{i}^\Omega (\theta)$}
            {$Index_j=Index_j+1$\\
           $ \min(Index_j)$}
        }}
        $point^*=points(\arg\_\min{\{Index(j)\}})$ \\

  %Find all $\{C(x,y)|(x,y)\in\varhexagon_0\}$ in orientation $\theta_0$\
%  Find $\min \{N(x,y) \}$\;
%
% Shift the center of $\Omega$ to $\text{Arg}\_\min \{N(x,y)$ \}\;
\Return list of all lattice points contained in the Minkowski sum
 \caption{Optimal translation given orientation}
 \label{alg:opt_pos}
\end{algorithm}
Since the orientation and location of the domain can be chosen independently, we can choose each of them separately. We will start by selecting the orientation as to minimize the {\em expected} number of intersecting hexagons, and then select the location as to minimize the number of intersections {\em for this orientation}. This may lead to a sub-optimal combination of orientation and translation, i.e., the desired location. However, this method allows us to provide a bound for the approximation ratio.

Algorithm \ref{alg:opt_pos}  determines the optimal positioning for a given  optimal orientation $\theta^*$ (see Lemma \ref{lem:exp}), in the sense of the expected number of lattice points for a random positioning.
The algorithm input are lattice points $\overline{x}_{mn}\in \Delta$ and the respective Minkowski sum's \[T_{\overline{x}_{mn}}^{\Omega}(\theta^*):=(R(\pi+\theta^*)\cdot\Omega) + \varhexagon_{\overline{x}_{mn}},\]
where again, $\theta^*$ is the optimal orientation of $\Omega$ found in Lemma \ref{lem:exp}.

Notice that if $\overline{x}_{mn} \notin \Delta$, then $\varhexagon_0 \cap T_{\overline{x}_{mn}}(\theta)= \emptyset$. Therefore, we can concentrate only on the points $\overline{x}_{mn} \in \Delta$.

Since the lattice is periodic, we can limit our discussion to the hexagon which located at the origin: $\varhexagon_0$ (it is sufficient to do the shifting in $ (R(\pi+\theta)\cdot\Omega)+\varhexagon_{(0+\varepsilon)}$, where  $\varepsilon \in \varhexagon_0$).

Algorithm \ref{alg:opt_pos} output gives the desired shifting.

\subsection{Combined orientation and location algorithm}
The algorithm we presented in Section \ref{sub:2D} finds the optimal orientation for a random location and then finds the optimal location for this orientation.
However, this does not guarantee that the combination of orientation and location is optimal. Indeed, consider a long rectangle of width 1. The Minkowski sum area is minimized when adding a hexagon oriented s.t.\ the upper edge if the hexagon is parallel to the upper edge of the rectangle. However a better orientation will be to rotate one of the regions by $\pi/2$. Hence we study an algorithm optimizing over both orientation and location, and this section is devoted to its presentation.

\longversion{
The group of motions $E(2)$, on a point $(x,y)\in \mathbb{R}^2$ is
\[
\left( \begin{array}{ccc} \cos\theta & -\sin\theta & a \\ \sin\theta & \cos\theta & b \\ 0 & 0 & 1  \end{array} \right)\left(\begin{array}{c} x \\ y \\ 1 \\ \end{array}\right)=\left(\begin{array}{c} x \cos\theta-y \sin\theta+a \\ x \sin\theta+y \cos\theta+b \\ 1 \\ \end{array}\right) \ .
\]
Let the group of motions be represented by a $3$-dimensional domain, i.e., for every point $(x,y) \in \Omega$, we will define the following function
\[
E(a,b,\theta):=\left(\begin{array}{c} x \cos\theta-y \sin\theta+a \\ x \sin\theta+y \cos\theta+b \\ \theta \\ \end{array}\right) \ .
\]
}

We denote the Minkowski sum as follows:
\begin{equation}
T_{ab}(\theta):=\left\{(x,y, \theta) \mid (x,y)\in\left(R(\pi+\theta)\cdot \Omega + \varhexagon_{ab}\right), \theta\in\left[0,\frac{\pi}{3}\right]\right\} ,
\label{eq:3.12}
\end{equation}
where $R(\pi+\theta)$ is a rotation.

\longversion{
Notice that instead of (\ref{eq:3.12}) we may use the definition
\[
T_{ab}(\theta):=\left\{(x,y, g(\theta)) \mid (x,y)\in\left(R(\pi+\theta)\cdot \Omega + \varhexagon_{ab}\right), \theta\in\left[0,\frac{\pi}{3}\right]\right\} \ ,
\label{eq:3.12alt}
\]
where $g(\theta)$ may be any injective function.
}

We define the Minkowski sum $T$ for each of the $n$ lattice points in the ball $B_{D+3}$ (where $D+3$ is the radius of the ball), which is the sequence
\[
\{ T_{\overline{x}_{mn} }^{\Omega}(\theta) \ |\overline{x}_{mn} \in B_{D+3}\cap \Lambda_h \} .
\]
In order to find the optimal placement and orientation of the convex region $\Omega$, we propose Algorithm \ref{alg:opt_pos_and_orient} : \newline
Each point $(x,y,\theta)$ represents a placement of $\Omega$ at the location $(x,y)$ and orientation $\theta$. Notice that a point $(x,y,\theta)$ is in the interior of a domain $T^\Omega_{\overline{x}_{mn} }(\theta)$ if and only if placing $\Omega$ at a location $(x,y)$ and orientation $\theta$ intersect with the hexagon that is centered at $\overline{x}_{mn} $.
%\textcolor{red}{We will note that the diameter of $T_{\overline{x}_{mn} }^{\Omega}(\theta)$ is smaller than $D+3$ and the unit hexagon is inscribed by the unit circle, %so we have no point from any $T_{(a_i,b_i)}$ that outside $B_{(D+3)}$, which intersects $\varhexagon_0$.
%so there is no $(x,y)$ such that$(x,y)\in T^{\Omega}_{\overline{x}_{mn}}(\theta)$ and $(x,y)\notin B_{D+3}$,  which intersects $\varhexagon_0$.
%Since the lattice is periodic in the lattice we can limit our discussion to the hexagon $\varhexagon_0$.}
In this case every Minkowski sum, respective to a lattice point $\overline{x}_{mn}$, generates a three-dimensional body by rotating the convex domain continuously $\frac{\pi}{3}$ radians and adding the hexagon.
 We will examine the domains inside $P_0:=\varhexagon_0 \times [0, \frac{\pi}{3}]$ (hexagonal prism).

 The hexagonal prism will be divided into different regions each of which is the intersection of a different number of $T_{\overline{x}_{mn} }^\Omega(\theta)$'s for lattice points $\overline{x}_{mn} \in B_{(D+3)}$.
\shortversion{Similarly to the translation optimizing version, it can be shown that by finding the intersections of the surfaces formed by the boundaries of the Minkowski sums, the optimal placement and angle can now be determined. The intersections can be found by solving  systems of polynomial equations.}.

\begin{algorithm}[H]
\SetAlgoLined
\KwResult{List of centers of unit discs covering $\Omega$. }
 Given a convex polygon $\Omega$\;
  \If{Diameter of minimum enclosing circle of $\Omega\leq 1$ }
 {\Return center of minimum enclosing circle}
 Create the hexagonal prism:$P_0:=\varhexagon_0 \times [0, \frac{\pi}{3}]$

$points=\emptyset$

 \For{$\overline{x}_{mn} \in \Delta$}{
  Find $T_{\overline{x}_{mn}}^\Omega (\theta)$
  }
  Define $y(i)\in \partial T_{i}^\Omega$
  \For {i=1:$|\overline{x}_{mn}|$}{
    \For {j=1:$|\overline{x}_{mn}|$}{
        \For {k=1:$|\overline{x}_{mn}|$}{
        \If {$y_i=y_j=y_k \quad \text{and} \quad T_{i}\neq T_j \neq T_k$ \text{and} $y_i\in P_0$ }{$points=[points, y_k]$ }}}}%
        \If {$\text{Corner point} \in  T_{i}\cap T_j$}{$points=[points, y_i]$}
        \If{$P_0\cap \partial T_{i}^\Omega (\theta)\neq \emptyset$}{$points=[points, y(i)]$}
  \For{j=1:$|{points}|$}{
    \For{i=1:$|\overline{x}_{mn}|$}{
        Index(j)=0 \\
        \If {$points(i)\in Int T_{i}^\Omega (\theta)$}
            {$Index(j)=Index(j)+1$\\
           $ \min(Index(j))$}
        }}
        $point^*=points(\arg\_\min{\{Index(j)\}})$ \\

  %Find all $\{C(x,y)|(x,y)\in\varhexagon_0\}$ in orientation $\theta_0$\
%  Find $\min \{N(x,y) \}$\;
%
% Shift the center of $\Omega$ to $\text{Arg}\_\min \{N(x,y)$ \}\;
\Return list of all lattice points contained in the Minkowski sum
 \caption{Optimal orientation and translation}
 \label{alg:opt_pos_and_orient}
\end{algorithm}

\shortversion{
  We will show that in order to examine all of these regions, we need only to choose for each region special points on its boundary.

In order to find a combination of a location and an orientation of $\Omega$ which leads to a minimal cover, we will do as we did in the previous section, with a slight change. Instead of inspecting the intersections points by the form $(x,y)$ which lie on the boundary of the sequence of Minkowski sums or the hexagon boundary, we will inspect the intersection points of the form $(x,y,\theta)$ which belong to the boundary surfaces of different $T_{\overline{x}_{mn}}^{\Omega}(\theta)$'s or surfaces of the hexagonal prism.
Note that $\theta$ can be given by $g(\theta)$ which is injective on $[0, \frac{\pi}{3}]$.

\subsection{Non-convex polygon location and orientation}

We will adapt the convex domain method (Algorithm \ref{alg:opt_pos_and_orient}) to the non-convex case. Let $\Gamma$ be a not necessarily convex polygon.
Informally, Algorithm \ref{alg:non-conv} follows the following steps:
\begin{enumerate}
\item Find a triangulation of $\Gamma$, denote each triangle by $\Gamma_i$.
\item Denote by $T_{ab}^i(\theta)$ the Minkowski sum of rotation by $\theta$ and translation by $(a,b) \in \Lambda_h$ of $\Gamma_i$. For each three surfaces intersection $(x,y)$ of the $T_{ab}^i(\theta)$, calculate
    \[|\{(a,b)| \exists i (x,y)\in int\, T^i_{ab}(\theta)\}|\]
    and find the minimum.
\end{enumerate}
Notice that the algorithm is almost identical to Algorithm \ref{alg:opt_pos_and_orient} applied to each triangle. However, when calculating the cardinality of each suspected point, we only count distinct translations, rather than distinct triangles. I.e. intersections with triangles belonging to the same translation are only counted as one.
%%%%%%%%%%%%%%%%%%%%%%%%%%%%%%%%%%%%%%%%%%

 \begin{algorithm}[H]
\SetAlgoLined
\KwResult{List of centers of unit discs covering $\Omega$. }
Given a non-convex polygon $\Gamma$ \newline
 \If{Diameter of minimum enclosing circle of $\Gamma\leq 1$ }
 {\Return center of minimum enclosing circle}
desired-lattice-points$=\emptyset$ \newline
Find a trinangulation of $\Gamma$ such that $\cup_{i=1}^n \Gamma_i=\Gamma$ \newline
 \While{$ab\in \Lambda_h \cap \Delta$ }{$\text{desired lattice points}=[\text{desired lattice points}, ab]$}
%    %\State $ab$
\For {i=1:n}{Find the Minkowaki sums $T_{\overline{x}_{mn}}^{\Gamma_i}(\theta)$}
Find all intersection points (same as in Algorithm \ref{alg:opt_pos_and_orient}).\newline
\For {each intersection point, $p$}{
Calculate $\left| \cup_{i=1}^n \left\{\overline{x}_{mn}|p\in T_{\overline{x}_{mn}}^{\Gamma_i}(\theta)\right\}\right|$.
}
\Return $\min_p \cup_{i=1}^n \left\{\overline{x}_{mn}|p\in T_{\overline{x}_{mn}}^{\Gamma_i}(\theta)\right\}$
 \caption{The Non-convex Algorithm}
 \label{alg:non-conv}
\end{algorithm}
%%%%%%%%%%%%%%%%%%%%%%%%%%%%%%%%%%%%%%%%%%%%%%
%\begin{algorithm}[H]
% \caption{My Caption\ldots}
%\SetAlgoLined
%\KwResult{A covering for a non-convex domain $\Gamma$ }
% %Find a triangulation of $\Gamma$. Denote each triangle by $\Gamma_i$ \\
% \For{ the index of $ |\Gamma_i|$}{
%  Find all Minkowski sums $T_{ab}^i(\theta)$ of the rotation and translation of $\Gamma_i$ by $(a,b)\in \Lambda_h$ (\textcolor{red}{does the minining here in the unit hexagon}) \\
%
% }
% \For { All intersections points of $3$ surfaces $\{N(x,y,\theta)\}$}{
% ssss
%  Find $\min \{N(x,y,\theta) \}$ } \;
%
% Shift the center of $\Omega$ to $\text{Arg}\_\min \{N(x,y,\theta)$ \}\;
% \caption{Non-convex Algorithm}
%\end{algorithm}
}
\section{Correctness and performance bounds}

\shortversion{
The canonical hexagon is the diamter 2 regular hexagon having the center of gravity at the origin and two vertices on the $y$ axis.  Directly calculating the support function of the canonical hexagon, one obtains \begin{equation}
h_{\varhexagon}(\varphi)=\max_{n\in\mathbb{Z}}\cos\left(\varphi-\frac{n\pi}{3}\right)\;.
\label{eq:0.7}
\end{equation}
}
\begin{lemma}
If we place $\Omega$ in orientation $\theta$ such that
\[
w_0(\theta)+w_0\left(\theta+\frac{\pi}{3}\right)+w_0\left(\theta+\frac{2\pi}{3}\right) \quad \left( 0 < \theta \leq  \pi \right) \
\]
is minimal, where $w_0$ is the width of $\Omega$, then the Minkowski sum $ (R(\pi)\cdot\Omega) +\varhexagon$ is minimized, and the expected number of covering hexagons is
\[
\left\lfloor \frac{2}{3\sqrt{3}}A+\frac{2}{3\sqrt{3}}\left(w_0(\theta)+w_0\left(\theta+\frac{\pi}{3}\right)+w_0\left(\theta+\frac{2\pi}{3}\right)\right)+1\ \right\rfloor.
\]
\label{lem:exp}
\end{lemma}
\begin{proof}
We will define the Minkowski sum of $\left(R(\pi)\cdot\Omega \right)+ \varhexagon$. The hexagon will be defined by (\ref{eq:0.7}). The convex domain $\Omega$ will be placed such that the hexagon and $\Omega$ intersect. We will denote the center of gravity of $\Omega$ by $r$.
Take an intersection point of $\Omega$ and $\varhexagon$. Then \[r+z_0=z_1+ 0 \Longleftrightarrow r=z_1+(-z_0) \ ,\]
where $z_0 \in \Omega$ and $z_1 \in \varhexagon$\longversion{ (as in Figure \ref{fig:int})}.
Denote the respective support function of $z_0$ by $h_1$ and in a similar way for $z_1$ by $h_{\varhexagon}$.
So, we will define the support function of the Minkowski sum as
\[
\hat{h}=h_{\varhexagon}(\varphi)+h_0(\varphi+\theta+\pi)\, .
\]

\longversion{
\begin{figure}[htp]
\centering
\input{../Images/Min_su_vec_thesis.tex}
	\caption{The support function $h_{\varhexagon}$ and $h_{0}(\theta)$ defines $\hat{h}$.}
\label{fig:int}
\end{figure}
}

By (\ref{eq:1.1}) the area of the Minkowski sum is:
\begin{align} \label{eq:2.2}
A(\varhexagon &+ R(\pi+\theta)\cdot\Omega) = \nonumber \\
& = \frac{1}{2}\int_0^{2 \pi} \left(\left(h_{\varhexagon}(\varphi)+(h_0(\varphi+\pi+\theta)\right)^2-\left(h'_{\varhexagon}(\varphi)+h'_{0}(\varphi+\pi+\theta)\right)^2\right) d \varphi \nonumber \\
& = A+   \frac{\sqrt{27}}{2}+\int_0^{2\pi}\left(h_0(\varphi+\pi+\theta)h_{\varhexagon}(\varphi)-h'_0(\varphi+\pi+\theta)h'_{\varhexagon}(\varphi) \right) d \varphi  \ .
\end{align}
 Denote the shifting factor in (\ref{eq:0.7}) for every interval by $c$. Thus we will get
\begin{eqnarray}
\int_0^{2\pi} & \left(h_0(\varphi+\pi+\theta)\cdot  \cos(\varphi+c)+h'_0(\varphi+\pi+\theta)\sin(\varphi+c)\right) d \varphi \nonumber \\ & =  \int_0^{2\pi}\left(h_0(\varphi+\pi+\theta)\cdot \sin(\varphi+c)\right)' d \varphi \ .
\label{eq:2.3}
\end{eqnarray}
The solution for (\ref{eq:2.3}) where $-\frac{\pi}{6}<\varphi \leq \frac{\pi}{6}$ gives the support function $h_{\varhexagon}=\cos{(\varphi)}$, is
\[
\int_{-\frac{\pi}{6}}^{\frac{\pi}{6}}\left( h_0(\varphi+\pi+\theta) \cdot \sin(\varphi)\right)' d \varphi = \frac{1}{2}\cdot h_0 \left(\frac{7\pi}{6}+\theta \right)+\frac{1}{2}h_0\left(\frac{5\pi}{6}+\theta\right) \ .
\]
Similar calculations can be conducted for the other five regimes.

Merging all the results gives:
 \begin{equation}
 \int_0^{2\pi}\left(h_0(\varphi+\pi+\theta)\cdot \sin(\varphi+c)\right)' d \varphi =\sum_{n=0}^5
 h_0\left(\frac{(2n+1)\pi}{6}+\theta\right) \ ,
 \label{eq:1.3}
\end{equation}
leading to
\begin{equation}
 f(\theta)=w_0(\theta)+w_0\left(\theta+\frac{\pi}{3}\right)+w_0\left(\theta+\frac{2\pi}{3}\right) \ .
\label{eq:2.4}
\end{equation}
Since we would like to determine the orientation of $\Omega$ such that the Minkowski sum is minimal, it is sufficient to minimize $f$. Thus the points that should be inspected are the critical points of (\ref{eq:2.4}) which are either zeroes or discontinuities of the derivative of $f$.
\end{proof}
\begin{theorem}
For any convex domain, $\Omega$, there exists a covering of  $\Omega$ with at most
\[
\left\lfloor \frac{2}{3\sqrt{3}}A+\frac{2}{3\sqrt{3}}\left(w_0(\theta)+w_0\left(\theta+\frac{\pi}{3}\right)+w_0\left(\theta+\frac{2\pi}{3}\right)\right)+1 \right\rfloor
  \label{eq:2}
\]
unit discs, for any $\theta$.
\end{theorem}
\begin{proof}
There exists a cover with at most the integer part of the expected cover size, where the expected size is given in Lemma \ref{lem:exp}.
\end{proof}
\longversion{
\begin{figure}[h]
\vskip -1.2in
\centering
\begin{minipage}[b]{0.45\linewidth}

\input{../Images/two_Min_sum1.tex}
\centering
\caption{The dotted area in the Minkowski sum of the inner square of the hexagon. The orientation of the square in this case is $\theta=0$ }
\label{fig:minipage1}
\end{minipage}
\quad
\begin{minipage}[b]{0.45\linewidth}
\vskip 0.8in
\input{../Images/two_min_sum2.tex}
\caption{The Minkowski sum of the same square in orientation $\theta=\frac{\pi}{2}$ and a hexagon, which give a bigger area than the square in orientation $\theta=0$}
\label{fig:minipage2}
\end{minipage}
\end{figure}
}

The definition of $f$ is valid for every $0 < \theta \leq  \pi$. Obviously, if the width is constant a minimization to $\theta$ is not possible. In this case we will get that
$
f(\theta)= \frac{3L}{\pi}$.
Since the  area of the fundamental region in the hexagonal lattice is $\frac{\sqrt{27}}{2}$, the ratio $ \frac{A\left((R(\pi)\cdot \Omega) + \varhexagon_0 \right)}{\sqrt{27}/2}$  determines the number of hexagonal lattice points which covers the convex domain $\Omega$.

\shortversion{\subsection{Correctness of the algorithm}}
\subsubsection{Correctness for separate orientation and location optimization}
The hexagon will be divided into different regions each of which is the intersection of a different number of elements in the series $(R(\pi+\theta)\cdot\Omega) + \varhexagon_{\overline{x}_{mn}}$.
 In order to examine all of these regions, we will choose for each region special points on its boundary, which will be defined as follows:
 %\textcolor{red}{maybe we need to make definitions more clear? especially def.8}
 %\textcolor{blue}{}
\begin{definition}
 Denote $Q(x,y)=\left\{T_{\overline{x}_{mn}}^\Omega (\theta) |(x,y) \in \text{Int} T_{\overline{x}_{mn}}^\Omega\right\}$ .
 \end{definition}
That is,  $Q(x,y)$ is the set of all regions including the point $(x,y)$ in their interior.
\begin{definition}
 Denote $C(x,y)=\bigcap_{T\in Q(x,y)}T$ .
 \end{definition}
That is, $C(x,y)$ is the intersection of all regions in $Q(x,y)$.
Notice that $\{C(x,y)|(x,y)\in\varhexagon_0\}$ is a partition of the unit hexagon into equivalence classes of points, which are convex\longversion{ (see Figure \ref{fig:co_nin_sum})}.
%\begin{figure}[h]
%\centering
%\input{../Images/Algorithm2_1.tex}
%\caption{Take all the lattice points which are contained in the square $\Delta$ and construct the respective Minkowski sums. Choose the intersection points between the Minkowski sums which are contained in $\varhexagon_0$ }
%\label{fig:co_nin_sum}
%\end{figure}
\begin{definition}
 Denote $N(x,y)=|Q(x,y)|$, i.e., the number of domains whose interior contains $(x,y)$.
 \end{definition}
The index of the intersections boundaries of the Minkowski sums are convex domains and finite, thus the intersection points can be easily calculated.
\begin{theorem}
\label{thm:fixed_angle_points}
Let $(x,y)\in\varhexagon_0$ then $\exists (\tilde{x}, \tilde{y})\in\varhexagon_0$ such that $N(\tilde{x},\tilde{y}) \leq N(x,y)$ and $(\tilde{x},\tilde{y})$ is either
 \begin{enumerate}[(i)]
  \item one of the intersection points of the pair of $T^{\Omega}_{\overline{x}_{mn}}$, whose boundaries lie inside the domain $\varhexagon_0$; or the intersection of a domain  $T^{\Omega}_{\overline{x}_{mn}}$ with the edges of the hexagon;
	%, or intersection of $T_{(\overline{x}_k)}^\Omega$ with the edges of the prism.
 \item  a corner point of $T_{\overline{x}_{mn}}^\Omega$ (which is contained in $\varhexagon_0$) or a corner point of $\varhexagon_0$.
 \end{enumerate}

 \end{theorem}

 \begin{proof}
Take a point $(x,y)\in\varhexagon_0$. Now, let $(x',y')$ be an arbitrary point in $\partial C(x,y)$. If $T\in Q(x',y')$ then $ \int T \cap C(x,y)\neq \emptyset$, since $\int T$ is an open set and thus $T\in Q(x,y)$. Thus $Q(x',y')\subseteq Q(x,y)$, and therefore $N(x',y')\leq N(x,y)$.

Now consider $C(x',y')$. Every point in $C(x',y')$ belongs to some edge of some $T^{\Omega}_{\overline{x}_{mn}}$ or of $\varhexagon_0$. If the closure $\overline{C}(x',y')$ contains a corner of a polygon or of the unit hexagon, $(\tilde x,\tilde y)$, we are done, as every $T_{\overline{x}_{mn}}^\Omega$ covering $(\tilde x,\tilde y)$ must also cover the interior of $C(x',y')$. Otherwise, since $\overline C(x',y')$ is a closed set, consisting of the intersection of boundaries of polygons and the unit hexagon, it must contain at least one point, $(\tilde x,\tilde y)$, of the intersection between two $T_{\overline{x}_{mn}}^\Omega$ or an intersection of one of the $T_{\overline{x}_{mn}}^\Omega$ and the unit hexagon. By the same argument, $N(\tilde x,\tilde y)\leq N(x',y')$.\end{proof}

 %The domains which in the sequence $\left\{T_{\overline{x}_k}^\Omega\right\}_{k=1}^n$ which intersects $\varhexagon_0$, divides $\varhexagon_0$ to sub domains which is defined by the boundaries of $T_{\overline{x}_k}^{\Omega}$ or the boundaries of $T^{\Omega}_{\overline{x}_k}$ and the boundary of $\varhexagon_0$. Denote those sub domains by $\left\{\varhexagon_{0_j}\right\}_{j=1}^m$.
%Taking a point $(x_1,y_1)\in int(\varhexagon_{0_j})$, then the cardinality of $N(x_1,y_1)$ is bigger then the boundary points. So, it is sufficient to examine the boundary points of $ \left\{\varhexagon_0{_j}\right\}_{j=1}^m$. Denote the vertices of $\varhexagon_{0_j}$ by $\left\{v_1\dots v_n\right\}$.  Without lost of generality we will explore the edge $\overline{v_1v_2}$.  If $(x_2,y_2)\in \overline{v_1v_2}$, then the cardinality of $N(x_2,y_2)$ is constant on the edge $\overline{v_1v_2}$. Now, since $N((x_2,y_2)) \geq N(v_1)$ (since $(x_2,y_2)$ may be in the interior of $T_{\overline{x_k}}$ but $v_2 \notin T _{\overline{x_k}}^\Omega$ ). So, it is sufficient to find the cardinality $N(x,y)$ for each of the vertices in $ \left\{\varhexagon_{0_j}\right\}_{j=1}^m$ ($\varhexagon_0= \bigcup_{j=1}^m \varhexagon_{0_j}$).
%
 %Since, this group is finite there exist $N(\tilde{x},\tilde{y}) \leq N(x,y)$.

\begin{theorem}
Algorithm \ref{alg:opt_pos} gives the optimal placement.
\end{theorem}

 \begin{proof}
 The algorithm checks all points which are of one of the types mentioned in the statement of Theorem \ref{thm:fixed_angle_points}. The correctness of the theorem follows immediately from the algorithm.\end{proof}
  So finally we will shift $\Omega$ to the point $(\tilde{x},\tilde{y})$ and get the desired location.

\subsubsection{Correctness for Combined orientation and location algorithm; and non-convex polygon}

\begin{definition}
 Denote $Q(x,y,\theta)=\left\{T_{\overline{x}_{mn}}^\Omega|(x,y,\theta) \in \int T_{\overline{x}_{mn}}^\Omega(\theta)\right\}$ .
 \end{definition}
\begin{definition}
 Denote $C(x,y,\theta)=\bigcap_{T\in Q(x,y,\theta)}T$ .
 \end{definition}
Notice that $\{C(x,y,\theta)|(x,y,\theta)\in P_0\}$ is a partition of the hexagonal prism into equivalence classes which are convex bodies.
\begin{definition}
 Denote $N(x,y,\theta)=|Q(x,y,\theta)|$, i.e., the number of domains whose interior contains $(x,y,\theta)$.
 \end{definition}
\begin{theorem}
 Let $(x,y, \theta)\in P_0$. Then $\exists (\tilde{x}, \tilde{y}, \tilde{\theta})$ such that $N(\tilde{x},\tilde{y}, \tilde{\theta}) \leq N(x,y, \theta)$ and $(\tilde{x},\tilde{y}, \tilde{\theta})$
is an intersection of three different surfaces from the set of surfaces of $\{\partial T_{mn}\}\cup\{\partial P_0\}$.
The desired points are either
\begin{enumerate}[I.]
\item An intersection point of the boundaries of three domains $T^\Omega_{\overline{x}_{mn}}$ whose boundaries reside inside the domain $\varhexagon_0 \times [0, \frac{\pi}{3}]$.
\item An  intersection of the boundaries of two domains in the sequence $T^\Omega_{\overline{x}_k}$, and the boundary of the domain $\varhexagon_0 \times [0, \frac{\pi}{3}]$.% As we did in (I.), we will count how many domains in the sequence $T(a_i,b_i)$ (($1 \leq i \leq n$) the point $p$ belongs to his interior.
 %\item If a domain has a corner point. As the other casesf
\end{enumerate}
%the intersection of three surfaces from the set including all surfaces of $T_{\overline{x}_{mn}}^\Omega$ and the surfaces of the hexagonal prism.
 %\end{enumerate}
 \label{thm:alg}
 \end{theorem}
\begin{proof}
 In a similar manner to what has been done for Algorithm \ref{alg:opt_pos}, we will do so for the $3D$ algorithm and the respective hexagonal prism.

 Take a point $(x,y,\theta)\in P_0$. Now, let $(x',y',\theta')$ be an arbitrary point in $\partial C(x,y,\theta)$. If $T\in Q(x',y',\theta')$ then $\int T\cap C(x,y,\theta)\neq \emptyset$, since $\int T$ is an open set. Thus $Q(x',y',\theta')\subseteq Q(x,y,\theta)$, and therefore $N(x',y',\theta')\leq N(x,y,\theta)$.

Now consider $C(x',y',\theta')$. Every point in $C(x',y',\theta')$ belongs to some surface of some $T^{\Omega}_{\overline{x}_{mn}}(\theta)$ or of $P_0$. If the closure $\overline{C}(x',y',\theta')$ contains a corner of the surface or of $P_0$, $(\tilde x,\tilde y,\tilde{\theta})$ we are done, as every $T_{\overline{x}_{mn}}^\Omega(\theta)$ covering $(\tilde x,\tilde y,\tilde{\theta})$ must also cover the interior of $C(x',y',\theta')$. Otherwise, since $\overline C(x',y',\theta')$ is a closed set, consisting of the intersection of boundaries of polygons and $P_0$, it must contain at least one point, $(\tilde x,\tilde y, \tilde{\theta})$, of intersection between two $T_{\overline{x}_{mn}}^\Omega$ or an intersection of one of the $T_{\overline{x}_{mn}}^\Omega$ and $P_0$. By the same argument, $N(\tilde x,\tilde y, \tilde{\theta})\leq N(x',y', \theta')$.\end{proof}
The bodies (each for a respective lattice point) which intersect the hexagonal prism divide it into sub-domains. The intersection between these domains which are contained in the hexagonal prism are plane lines or vertices. In the case of a line, since it is contained in the hexagonal prism, instead of studies all the points on the line we can take the respective end points, since it will not change in how many bodies this point belongs to the interior (as has been done in the 2-D algorithm). If the edge of the sub domain which is contained in the hexagonal prism is a plane, in a similar way it in bounded by lines, these lines have the same $N(x,y,\theta)$ for every point on the line. So it is sufficient to take the vertices of these lines.

\begin{theorem}
Algorithm \ref{alg:opt_pos_and_orient} gives the optimal placement.
\end{theorem}
\begin{proof}
 The algorithm checks all points which are of one of the types mentioned in Theorem~\ref{thm:alg}, each case leads to a point. The correctness of the theorem follows immediately from the algorithm.\end{proof}

 So, finally, we will shift Ω to the point $(x,y,\theta)$ and get the desired location.

 %The bodies (each for a respective lattice point) which intersect the hexagonal prism divide it to sub-domains. The intersection between these domains which contained in the hexagonal prism are planes lines or vertices. In the case of a line since it contained in the hexagonal prism, instead of studies all the points on the line we can take the respective end points, since it will not change in how many bodies this point belongs to the interior (as has been done in the 2-D algorithm). If the edge of the sub domain which contained in the hexagonal prism is a plane, in a similar way it in bounded by lines, these lines have the same $N(x,y,\theta)$ for every point on the line. So it is sufficient to take the vertices of these lines.

\subsubsection{Correctness for non-convex polygon placement}

%We will adapt our method to convex domain (the 3D algorithm) which find a point $(x,y)$ in orientation $\theta$ to a non-convex polygon $\Gamma$.
%The desired algorithm for $\Gamma$ has the following steps:
%\begin{enumerate}
%\item Find a triangulation of $\Gamma$, denote each triangle by $\Gamma_i$.
%\item Denote by $T_{ab}^i(\theta)$ the Minkowski sum of the rotation and translation of $\Gamma_i$ by \linebreak$(a,b)\in \Lambda_h$. For each $3$ surface intersection $(x,y)$ of the $T_{ab}^i(\theta)$, calculate $$|\{(a,b)| \exists i (x,y)\in int\, T^i_{ab}(\theta)\}|$$ and find the minimum.
%\end{enumerate}
The Algorithm in this case is similar to the algorithm which was obtained in Lemma \ref{thm:alg}. In order to simplify the structure of surfaces, it works only with convex domains (the Minkowski sum of non-convex domain is not necessarily composed of simple curves). However it only counts once the contribution from all triangles of the original domain and thus produces the desired result.%In order to simplify the structure of surfaces
%%%%%%%%%%%%%%%%%%%%%%%%%%%%%%%%%%%%%%%%%%

 %\begin{algorithm}[H]
%\SetAlgoLined
%\KwResult{List of centers of unit discs covering $\Omega$. }
%Given a non-convex polygon $\Gamma$ \newline
% \If{Diameter of minimum enclosing circle of $\Gamma\leq 1$ }
% {\Return center of minimum enclosing circle}
%desired-lattice-points$=\emptyset$ \newline
%Find a trinangulation of $\Gamma$ such that $\cup_{i=1}^n \Gamma_i=\Gamma$ \newline
% \While{$ab\in \Lambda_h \cap \Delta$ }{$\text{desired lattice points}=[\text{desired lattice points}, ab]$}
%%    %\State $ab$
%\For {i=1:n}{Find the Minkowaki sums $T_{\overline{x}_{mn}}^{\Gamma_i}(\theta)$}
%Find all intersection points (same as Algorithm $3D$).\newline
%\For {each intersection point, $p$}{
%Calculate $\left| \cup_{i=1}^n \left\{\overline{x}_{mn}|p\in T_{\overline{x}_{mn}}^{\Gamma_i}(\theta)\right\}\right|$.
%}
%\Return $\min_p \cup_{i=1}^n \left\{\overline{x}_{mn}|p\in T_{\overline{x}_{mn}}^{\Gamma_i}(\theta)\right\}$
% \caption{The Non-convex 3D Algorithm}
%\end{algorithm}
\shortversion{\subsection{Performance bounds}\label{sec:performance}}
\begin{theorem}
Let $n_\mathrm{opt}$ be the minimum number of unit discs necessary to cover
a convex domain $\Omega$ with area $A$ and circumference $L$, then the  algorithms give an approximation ratio of $ 1+\frac{8}{\pi\sqrt{3}}+o(1)$.
\end{theorem}
\begin{proof}
The algorithms find an optimal (placement or angle and placement, respectively) covering using the hexagonal lattice. Thus, from Theorem \ref{thm:toth}, it follows that either algorithms give a covering using $n\leq\frac{2}{3\sqrt{3}}A+\frac{2}{\pi\sqrt3}L+1$.
Therefore, by Theorem \ref{thm:lower_bnd},
\[\frac{n}{n_\mathrm{opt}}\leq \frac{2A/\sqrt{27}+2L/(\pi\sqrt{3})+1} {\max\{2A/\sqrt{27},L/4\}-C}\leq 1+\frac{8}{\pi\sqrt{3}}+o(1)\;, \]
and the theorem follows.
\end{proof}

We consider two asymptotic scenarios.
In the case of fat regions, having $L=o(A)$ we have
\begin{conclusion}
Let $\Omega$ be a fat convex polygon. The algorithm is asymptotically optimal. That is, the covering it produces uses $n= (1+o(1))n_\mathrm{opt}$ discs where $ n_\mathrm{opt}$ is the minimum number of unit discs needed to cover $\Omega$.
\end{conclusion}

On the other hand, if $A = \alpha L$ for some constant $\alpha$.
Let $D$ be the diameter of $\Omega$. Then we have that $2D\le L\le \pi D$. We also have that if the width of $\Omega$ in the direction perpendicular to the diameter is $\beta$ then the area of $\omega$ must
be at least $D \cdot \beta/2$ (as one can built two triangles with the diameter as the base and a sum of heights of $\beta$). Thus, we has $\beta <2pi \cdot  A$.

Now use Lemma 4 and choose $\theta$ to be in the direction perpendicular to the diameter. one has $w(\theta)\le \beta$, and also, since $\Omega$ can be inscribed in a rectangle of sdes $D$ and $\beta$, we have $w(\theta+\frac{\pi}{3})<D\cos(\pi/3)+C$, and the same for $w(\theta+\frac{2\pi}{3})$. Thus, we have that
$$n\leq \lfloor \frac{2A}{3\sqrt{3}}+\frac{L}{3}+C\rfloor\;.$$

\section{Computational Complexity}
\subsection{Complexity of Separate orientation and location optimization algorithm}
%\longversion{\subsection{2D Algorithm}}
\begin{theorem}
The optimal location and translation of $\Omega$ can be found in $O(D^3N^2)$ operations, where $D$ is the diameter of $\Omega$
and $N$ is the number of sides of $\Omega$.
\end{theorem}
\begin{proof}
%The worse case scenario where obtained intersections between two Minkowski domains.
The number of domains is the number of lattice points in $\Delta$. Due to the properties of the Minkowski sum with the hexagon, the number of edges of each polygon is at most $N+6$.
Thus, the number of intersection points between polygons is at most $O(D^2N)$. For each such intersection point, one needs to examine how many other polygons contain it, requiring $O(ND)$ operations. Thus, the total time complexity is $O(D^3N^2)$.
\end{proof}

\subsection{Complexity of the combined orientation and location optimization algorithm}

\begin{theorem}
The optimal location and translation of $\Omega$ can be found in $O(D^8N^4)$ operations, where $D$ is the diameter of $\Omega$
and $N$ is the number of sides of $\Omega$.
\end{theorem}
\begin{proof}

The algorithm goes over all choices of three surfaces and checks the intersection points. The number of surfaces is at most $(2D+7)^2(N+6)$
plus 6 surfaces of the hexagonal prism. There are at most 6 intersection points for each choice of three surfaces. Thus, the number of points that
is $O(D^6N^3)$.
For each
such intersection point, the algorithm needs to examine how many of the $(2D+7)^2$ regimes contains the point, requiring
examining $N+6$ inequalities.
\end{proof}

\subsection{Complexity non-convex algorithm}
\begin{theorem}
The optimal location and translation of $\Gamma$ can be found in $O(D^8N^4)$ operations, where $D$ is the diameter of $\Gamma$
and $N$ is the number of sides of $\Omega$.
\end{theorem}
\begin{proof}
Producing a trianulation of a polygon requires $O(N)$ operations \cite{Chazelle1991}.
Given a  triangulation by $N-2$ triangles of $\Gamma$. The number of surfaces respective to $T_{\overline{x}_{mn}}^{\Gamma_i}(\theta)$ is (in a similar way to the 3D algorithm) at most $O(D^2 N)$.
The algorithm goes over all choices of three surfaces and checks the intersection points. There are at most 6 intersection points for each choice of three surfaces. Thus the number of points that
is $O(D^6N^3)$.
For each
such intersection point the algorithm need to examine how many of the $O(D^2N)$ regimes contains the point.
\end{proof}

\bibliographystyle{plainurl}
\bibliography{manuscript}
\end{document}

%% file: Min_Sum.tex
\definecolor{uuuuuu}{rgb}{0.26666666666666666,0.26666666666666666,0.26666666666666666}
\definecolor{xdxdff}{rgb}{0.49019607843137253,0.49019607843137253,1.}
\definecolor{zzttqq}{rgb}{0.6,0.2,0.}
\begin{tikzpicture}[line cap=round,line join=round,>=triangle 45,x=1.0cm,y=1.0cm]
\clip(-4.307602594682851,-2.4213146427352235) rectangle (5.6239283128855275,2.759301114046327);
\fill[line width=0.8pt,color=zzttqq,fill=zzttqq,fill opacity=0.10000000149011612] (0.,-1.) -- (0.8660254037844386,-0.5) -- (0.8660254037844386,0.5) -- (0.,1.) -- (-0.8660254037844387,0.5) -- (-0.8660254037844392,-0.5) -- cycle;
\fill[line width=0.8pt,fill=black,fill opacity=0.10000000149011612] (2.4230446235710117,1.812353841640518) -- (0.8633540349285211,0.5207819762621475) -- (2.606516685561288,-0.43176592025849037) -- (4.155095431139892,0.812353841640518) -- (4.2734755044724055,1.7590942417389774) -- cycle;
\fill[line width=0.8pt,dash pattern=on 1pt off 1pt,fill=black,fill opacity=0.10000000149011612] (0.10903800157187488,2.284734516011137) -- (2.7009782931640083,0.8489727721950916) -- (2.7009782931640083,-0.15120251915157853) -- (1.1585354673465895,-1.4593852458878707) -- (0.2925100635621508,-1.9593852458878707) -- (-1.5579208173392427,-1.9061256459863296) -- (-2.423946221123682,-1.4061256459863292) -- (-2.4239462211236815,-0.4061256459863293) -- (-2.3055661477911653,0.5406147541121301) -- (-0.7569874022125638,1.7847345160111379) -- cycle;
\draw [line width=0.8pt,color=zzttqq] (0.,-1.)-- (0.8660254037844386,-0.5);
\draw [line width=0.8pt,color=zzttqq] (0.8660254037844386,-0.5)-- (0.8660254037844386,0.5);
\draw [line width=0.8pt,color=zzttqq] (0.8660254037844386,0.5)-- (0.,1.);
\draw [line width=0.8pt,color=zzttqq] (0.,1.)-- (-0.8660254037844387,0.5);
\draw [line width=0.8pt,color=zzttqq] (-0.8660254037844387,0.5)-- (-0.8660254037844392,-0.5);
\draw [line width=0.8pt,color=zzttqq] (-0.8660254037844392,-0.5)-- (0.,-1.);
\draw [line width=0.8pt] (2.4230446235710117,1.812353841640518)-- (0.8633540349285211,0.5207819762621475);
\draw [line width=0.8pt] (0.8633540349285211,0.5207819762621475)-- (2.606516685561288,-0.43176592025849037);
\draw [line width=0.8pt] (2.606516685561288,-0.43176592025849037)-- (4.155095431139892,0.812353841640518);
\draw [line width=0.8pt] (4.155095431139892,0.812353841640518)-- (4.2734755044724055,1.7590942417389774);
\draw [line width=0.8pt] (4.2734755044724055,1.7590942417389774)-- (2.4230446235710117,1.812353841640518);
\draw [line width=0.8pt,dash pattern=on 1pt off 1pt] (0.10903800157187488,2.284734516011137)-- (2.7009782931640083,0.8489727721950916);
\draw [line width=0.8pt,dash pattern=on 1pt off 1pt] (2.7009782931640083,0.8489727721950916)-- (2.7009782931640083,-0.15120251915157853);
\draw [line width=0.8pt,dash pattern=on 1pt off 1pt] (2.7009782931640083,-0.15120251915157853)-- (1.1585354673465895,-1.4593852458878707);
\draw [line width=0.8pt,dash pattern=on 1pt off 1pt] (1.1585354673465895,-1.4593852458878707)-- (0.2925100635621508,-1.9593852458878707);
\draw [line width=0.8pt,dash pattern=on 1pt off 1pt] (0.2925100635621508,-1.9593852458878707)-- (-1.5579208173392427,-1.9061256459863296);
\draw [line width=0.8pt,dash pattern=on 1pt off 1pt] (-1.5579208173392427,-1.9061256459863296)-- (-2.423946221123682,-1.4061256459863292);
\draw [line width=0.8pt,dash pattern=on 1pt off 1pt] (-2.423946221123682,-1.4061256459863292)-- (-2.4239462211236815,-0.4061256459863293);
\draw [line width=0.8pt,dash pattern=on 1pt off 1pt] (-2.4239462211236815,-0.4061256459863293)-- (-2.3055661477911653,0.5406147541121301);
\draw [line width=0.8pt,dash pattern=on 1pt off 1pt] (-2.3055661477911653,0.5406147541121301)-- (-0.7569874022125638,1.7847345160111379);
\draw [line width=0.8pt,dash pattern=on 1pt off 1pt] (-0.7569874022125638,1.7847345160111379)-- (0.10903800157187488,2.284734516011137);
\draw [->,line width=0.8pt] (-0.0042577329546052145,0.) -- (0.8660254037844386,0.5);
\draw [->,line width=0.8pt] (2.721985866223757,0.8310775939734412) -- (0.8660254037844386,0.5);
\draw [->,line width=0.8pt] (-0.0042577329546052145,0.) -- (2.721985866223757,0.8310775939734412);
\begin{scriptsize}
\draw [fill=xdxdff] (-0.0042577329546052145,0.) circle (1.5pt);
\draw [fill=uuuuuu] (2.721985866223757,0.8310775939734412) circle (1.5pt);
\draw [fill=xdxdff] (0.8660254037844386,0.5) circle (1.5pt);
\end{scriptsize}
\end{tikzpicture}

%% file: intersect.tex
\definecolor{uququq}{rgb}{0.25098039215686274,0.25098039215686274,0.25098039215686274}
\definecolor{uuuuuu}{rgb}{0.26666666666666666,0.26666666666666666,0.26666666666666666}
\definecolor{qqqqff}{rgb}{0.,0.,1.}
\definecolor{zzttqq}{rgb}{0.6,0.2,0.}
\begin{tikzpicture}[line cap=round,line join=round,>=triangle 45,x=1.0cm,y=1.0cm,scale=0.8]
\clip(-10.98529679724845,-5.623424654174697) rectangle (11.874644725554717,5.879170790083774);
\fill[line width=0.8pt,fill=black,fill opacity=0.10000000149011612] (0.,-1.) -- (0.8660254037844386,-0.5) -- (0.8660254037844386,0.5) -- (0.,1.) -- (-0.8660254037844387,0.5) -- (-0.8660254037844392,-0.5) -- cycle;
\fill[line width=0.8pt,color=zzttqq,fill=zzttqq,fill opacity=0.10000000149011612] (0.10903800157187488,2.284734516011137) -- (2.7009782931640083,0.8489727721950916) -- (2.7009782931640083,-0.15120251915157853) -- (1.1585354673465895,-1.4593852458878707) -- (0.2925100635621508,-1.9593852458878707) -- (-1.5579208173392427,-1.9061256459863296) -- (-2.423946221123682,-1.4061256459863292) -- (-2.4239462211236815,-0.4061256459863293) -- (-2.3055661477911653,0.5406147541121301) -- (-0.7569874022125638,1.7847345160111379) -- cycle;
\fill[line width=0.8pt,color=zzttqq,fill=zzttqq,fill opacity=0.10000000149011612] (-3.3369438781143117,-0.7118007455449216) -- (-0.7450035865221785,-2.147562489360967) -- (-0.7450035865221785,-3.1477377807076374) -- (-2.2874464123395972,-4.455920507443929) -- (-3.153471816124036,-4.955920507443929) -- (-5.0039026970254294,-4.902660907542389) -- (-5.869928100809869,-4.402660907542388) -- (-5.869928100809869,-3.402660907542388) -- (-5.751548027477352,-2.4559205074439285) -- (-4.202969281898751,-1.211800745544921) -- cycle;
\fill[line width=0.8pt,color=zzttqq,fill=zzttqq,fill opacity=0.10000000149011612] (-1.6087380393548398,-0.7195401930625681) -- (0.9689274855951311,-2.176946921938492) -- (0.9689274855951311,-3.177122213285162) -- (-0.5735153402222877,-4.485304940021454) -- (-1.4395407440067265,-4.985304940021454) -- (-3.153471816124036,-4.955920507443929) -- (-4.165156253417765,-4.344050948617536) -- (-4.147646539622813,-3.3459972623052985) -- (-4.025078543058153,-2.4530018587627698) -- (-2.423946221123682,-1.2175404465530901) -- cycle;
\fill[line width=0.8pt,color=zzttqq,fill=zzttqq,fill opacity=0.10000000149011612] (-2.490961998428125,0.7847345160111372) -- (0.10721391255041136,-0.7370499068575196) -- (0.08914206919101209,-1.6795180765652522) -- (-1.4686603289952274,-2.9257641312264613) -- (-2.326636304947853,-3.4160361174851044) -- (-4.147646539622812,-3.345997262305295) -- (-4.988112801780487,-2.873234989841607) -- (-5.059328979528317,-1.8997954801728507) -- (-4.905566147791165,-0.9593852458878699) -- (-3.356987402212564,0.28473451601113786) -- cycle;
\fill[line width=0.8pt,color=zzttqq,fill=zzttqq,fill opacity=0.10000000149011612] (0.17412999022963405,-2.906125645986329) -- (1.0401553940140733,-3.4061256459863296) -- (2.8905862749154667,-3.4593852458878707) -- (3.7566116786999055,-2.9593852458878707) -- (5.299054504517324,-1.6512025191515785) -- (5.299054504517324,-0.6510272278049084) -- (2.707114212925191,0.7847345160111372) -- (1.875695005840517,0.3485523484294755) -- (0.2823110504999268,-0.9646761861918896) -- (0.17725276773021753,-1.77012302075966) -- cycle;
\fill[line width=0.8pt,color=zzttqq,fill=zzttqq,fill opacity=0.10000000149011612] (-2.489038209781441,0.7847345160111372) -- (-3.35506361356588,0.28473451601113786) -- (-4.903642359144481,-0.9593852458878699) -- (-5.022022432476998,-1.9061256459863292) -- (-5.022022432476998,-2.906125645986329) -- (-4.11262711203291,-3.381016689895201) -- (-2.3055661477911653,-3.4593852458878707) -- (-1.4511506152002758,-3.0259196941335835) -- (0.,-1.595025882810145) -- (0.10721391255041136,-0.7370499068575196) -- cycle;
\fill[line width=0.8pt,color=zzttqq,fill=zzttqq,fill opacity=0.10000000149011612] (0.07219448496050827,0.3835717760193786) -- (0.08970419875545982,-0.649501337882762) -- (-1.4395407440067265,-1.9593852458878707) -- (-2.3055661477911653,-2.4593852458878707) -- (-4.140556581794558,-2.3905514570939643) -- (-5.022022432476998,-1.9061256459863292) -- (-5.022022432476998,-0.9061256459863293) -- (-4.970603087985536,0.) -- (-3.35506361356588,1.2847345160111379) -- (-2.489038209781441,1.7847345160111372) -- cycle;
\fill[line width=0.8pt,color=zzttqq,fill=zzttqq,fill opacity=0.10000000149011612] (-0.7569874022125642,1.7847345160111379) -- (1.834952889379569,0.3489727721950923) -- (1.834952889379569,-0.6512025191515778) -- (0.29251006356215037,-1.95938524588787) -- (-0.5735153402222883,-2.4593852458878698) -- (-2.361655732537756,-2.3654532897880123) -- (-3.289670563670188,-1.8401618759394662) -- (-3.307180277465139,-0.8596179034221804) -- (-3.2196317084903816,0.) -- (-1.625765739537775,1.31246769611255) -- cycle;
\fill[line width=0.8pt,color=zzttqq,fill=zzttqq,fill opacity=0.10000000149011612] (0.057239039181291695,1.3337074691501623) -- (0.9750634053563134,1.7847345160111372) -- (3.567003696948447,0.3489727721950916) -- (3.567003696948447,-0.6512025191515786) -- (2.024560871131028,-1.9593852458878707) -- (1.1585354673465895,-2.4593852458878707) -- (-0.8630928130169472,-2.444173588806049) -- (-1.5386991841750335,-1.8576715897344178) -- (-1.5912283255598882,-0.7895790482423742) -- (-1.4511506152002758,0.15594549668500865) -- cycle;
\fill[line width=0.8pt,color=zzttqq,fill=zzttqq,fill opacity=0.10000000149011612] (1.8320101678787537,1.3100432655576164) -- (2.707114212925191,1.7847345160111372) -- (5.299054504517324,0.3489727721950916) -- (5.299054504517324,-0.6512025191515786) -- (3.7566116786999055,-1.9593852458878707) -- (2.8905862749154667,-2.4593852458878707) -- (1.1585354673465895,-2.4593852458878707) -- (0.19476248152516906,-1.8576715897344178) -- (0.1937268067341193,-0.8881517759828096) -- (0.2998207642948783,0.) -- cycle;
\fill[line width=0.8pt,color=zzttqq,fill=zzttqq,fill opacity=0.10000000149011612] (0.9750634053563132,1.7847345160111372) -- (1.841088809140752,2.284734516011137) -- (4.4191667564150965,0.8363717958823323) -- (4.414603506108491,-0.159229351624119) -- (2.873748692152755,-1.437438458655581) -- (2.0245608711310306,-1.9593852458878702) -- (-0.6281940668375534,-1.4199287448606295) -- (-0.6281940668375534,-0.36934591716353743) -- (-0.5581552116577472,0.5236494863789909) -- cycle;
\fill[line width=0.8pt,color=zzttqq,fill=zzttqq,fill opacity=0.10000000149011612] (1.8410888091407522,3.284734516011138) -- (2.707114212925191,3.784734516011137) -- (5.299054504517324,2.3489727721950917) -- (5.299054504517324,1.3487974808484215) -- (3.6967052405154774,0.) -- (2.8905862749154667,-0.45938524588787066) -- (0.8660254037844386,-0.5) -- (0.26480133670497524,0.17345521047996018) -- (0.1741299902296345,1.0938743540136708) -- (0.2925100635621507,2.0406147541121302) -- cycle;
\fill[line width=0.8pt,color=zzttqq,fill=zzttqq,fill opacity=0.10000000149011612] (0.9750634053563134,3.784734516011137) -- (3.5391178163609136,2.3796791486438535) -- (3.521608102565962,1.3290963209467614) -- (1.9632435748152748,0.) -- (1.1056846271618606,-0.4951921917651505) -- (-0.6281940668375534,-0.38685563095848896) -- (-1.5579208173392434,0.09387435401367084) -- (-1.557920817339243,1.0938743540136708) -- (-1.4395407440067267,2.0406147541121302) -- (0.10903800157187482,3.284734516011138) -- cycle;
\fill[line width=0.8pt,color=zzttqq,fill=zzttqq,fill opacity=0.10000000149011612] (-0.6281940668375534,-0.38685563095848896) -- (0.31733047808982984,0.) -- (1.8349528893795697,1.3487974808484215) -- (1.8349528893795697,2.3489727721950917) -- (-0.7569874022125638,3.784734516011137) -- (-1.6230128059970024,3.284734516011138) -- (-3.171591551575604,2.0406147541121302) -- (-3.2899716249081203,1.0938743540136708) -- (-3.2899716249081203,0.09387435401367084) -- (-2.423946221123681,-0.4061256459863296) -- cycle;
\fill[line width=0.8pt,color=zzttqq,fill=zzttqq,fill opacity=0.10000000149011612] (-2.489038209781441,3.784734516011137) -- (0.08970419875545982,2.309640293464047) -- (0.10721391255041136,1.346606034741713) -- (-1.4336409014053244,0.) -- (-2.2741055284164906,-0.40451772231303323) -- (-4.155997028692559,-0.4061256459863296) -- (-5.022022432476998,0.09387435401367084) -- (-5.022022432476998,1.0938743540136708) -- (-4.903642359144481,2.0406147541121302) -- (-3.35506361356588,3.284734516011138) -- cycle;
\fill[line width=0.8pt,color=zzttqq,fill=zzttqq,fill opacity=0.10000000149011612] (-3.3550636135658793,5.284734516011137) -- (-0.7631233219737461,3.8489727721950917) -- (-0.7631233219737461,2.8487974808484213) -- (-2.305566147791165,1.5406147541121293) -- (-3.1715915515756037,1.0406147541121293) -- (-5.022022432476997,1.0938743540136704) -- (-5.888047836261436,1.5938743540136708) -- (-5.8880478362614355,2.593874354013671) -- (-5.76966776292892,3.5406147541121302) -- (-4.221089017350318,4.784734516011138) -- cycle;
\fill[line width=0.8pt,color=zzttqq,fill=zzttqq,fill opacity=0.10000000149011612] (-4.025078543058153,3.605359114290461) -- (-4.147646539622813,2.659834569363078) -- (-4.147646539622813,1.6617808830508407) -- (-3.237141422285333,1.118979755407343) -- (-1.4395407440067265,1.0406147541121293) -- (-0.5735153402222877,1.5406147541121293) -- (0.9689274855951311,2.8487974808484213) -- (0.9689274855951311,3.8489727721950917) -- (-1.6230128059970024,5.284734516011137) -- (-2.508338561978974,4.815705410463189) -- cycle;
\fill[line width=0.8pt,color=zzttqq,fill=zzttqq,fill opacity=0.10000000149011612] (2.7009782931640083,3.8489727721950917) -- (2.7009782931640083,2.8487974808484213) -- (1.1585354673465895,1.5406147541121293) -- (0.2925100635621508,1.0406147541121293) -- (-1.5386991841750335,1.153999182997246) -- (-2.4316945877175624,1.6617808830508407) -- (-2.466714015307465,2.6773442831580296) -- (-2.326636304947853,3.570339686700558) -- (-0.7752081608871471,4.815107387783365) -- (0.10903800157187488,5.284734516011137) -- cycle;
\fill[line width=0.8pt,color=zzttqq,fill=zzttqq,fill opacity=0.10000000149011612] (0.9832739804525663,4.800466194855695) -- (1.841088809140752,5.284734516011137) -- (4.4330291007328855,3.8489727721950917) -- (4.4330291007328855,2.8487974808484213) -- (2.8905862749154667,1.5406147541121293) -- (2.024560871131028,1.0406147541121293) -- (0.14223334014031444,1.1364894692022947) -- (-0.6457037806325049,1.609251741665986) -- (-0.6982329220173595,2.659834569363078) -- (-0.6106843530426018,3.6403785418803642) -- cycle;
\fill[line width=0.8pt,color=zzttqq,fill=zzttqq,fill opacity=0.10000000149011612] (3.5731396167096294,5.284734516011137) -- (6.165079908301763,3.8489727721950917) -- (6.165079908301763,2.8487974808484213) -- (4.622637082484344,1.5406147541121293) -- (3.756611678699905,1.0406147541121293) -- (1.7320508075688772,1.) -- (0.9826996022979885,1.679290596845792) -- (1.0352287436828431,2.6773442831580296) -- (1.1577967402475038,3.6228688280854127) -- (2.7535137327049597,4.779312071054633) -- cycle;
\fill[line width=0.8pt,color=zzttqq,fill=zzttqq,fill opacity=0.10000000149011612] (1.8349528893795688,-0.6512025191515778) -- (4.484642361288297,-2.1728464380435457) -- (4.4330291007328855,-3.1512025191515787) -- (2.8905862749154667,-4.459385245887871) -- (2.024560871131028,-4.959385245887871) -- (0.1741299902296345,-4.90612564598633) -- (-0.6450178235220853,-4.414089803797341) -- (-0.7149822736075474,-3.345997262305298) -- (-0.5153403947258817,-2.4257979254454276) -- (0.9466734159221789,-1.227321893116162) -- cycle;
\fill[line width=0.8pt,color=zzttqq,fill=zzttqq,fill opacity=0.10000000149011612] (3.5670036969484467,-0.6512025191515785) -- (6.165079908301763,-2.1510272278049083) -- (6.165079908301763,-3.1512025191515787) -- (4.622637082484344,-4.459385245887871) -- (3.756611678699905,-4.959385245887871) -- (1.7334609450754725,-4.951006729428215) -- (1.1040934286144688,-4.39658009000239) -- (1.0401553940140729,-3.406125645986329) -- (1.1585354673465897,-2.4593852458878707) -- (2.721564543722245,-1.1603592511204144) -- cycle;
\draw [line width=0.8pt] (0.,-1.)-- (0.8660254037844386,-0.5);
\draw [line width=0.8pt] (0.8660254037844386,-0.5)-- (0.8660254037844386,0.5);
\draw [line width=0.8pt] (0.8660254037844386,0.5)-- (0.,1.);
\draw [line width=0.8pt] (0.,1.)-- (-0.8660254037844387,0.5);
\draw [line width=0.8pt] (-0.8660254037844387,0.5)-- (-0.8660254037844392,-0.5);
\draw [line width=0.8pt] (-0.8660254037844392,-0.5)-- (0.,-1.);
\draw [line width=0.8pt,color=zzttqq] (0.10903800157187488,2.284734516011137)-- (2.7009782931640083,0.8489727721950916);
\draw [line width=0.8pt,color=zzttqq] (2.7009782931640083,0.8489727721950916)-- (2.7009782931640083,-0.15120251915157853);
\draw [line width=0.8pt,color=zzttqq] (2.7009782931640083,-0.15120251915157853)-- (1.1585354673465895,-1.4593852458878707);
\draw [line width=0.8pt,color=zzttqq] (1.1585354673465895,-1.4593852458878707)-- (0.2925100635621508,-1.9593852458878707);
\draw [line width=0.8pt,color=zzttqq] (0.2925100635621508,-1.9593852458878707)-- (-1.5579208173392427,-1.9061256459863296);
\draw [line width=0.8pt,color=zzttqq] (-1.5579208173392427,-1.9061256459863296)-- (-2.423946221123682,-1.4061256459863292);
\draw [line width=0.8pt,color=zzttqq] (-2.423946221123682,-1.4061256459863292)-- (-2.4239462211236815,-0.4061256459863293);
\draw [line width=0.8pt,color=zzttqq] (-2.4239462211236815,-0.4061256459863293)-- (-2.3055661477911653,0.5406147541121301);
\draw [line width=0.8pt,color=zzttqq] (-2.3055661477911653,0.5406147541121301)-- (-0.7569874022125638,1.7847345160111379);
\draw [line width=0.8pt,color=zzttqq] (-0.7569874022125638,1.7847345160111379)-- (0.10903800157187488,2.284734516011137);
\draw [line width=0.8pt] (-4.025078543058153,4.025592245369296)-- (-4.025078543058153,-3.0259196941335835);
\draw [line width=0.8pt] (-4.025078543058153,-3.0259196941335835)-- (4.0013742605476335,-3.0259196941335835);
\draw [line width=0.8pt] (4.0013742605476335,-3.0259196941335835)-- (4.0013742605476335,3.9972265090214774);
\draw [line width=0.8pt] (4.0013742605476335,3.9972265090214774)-- (-4.025078543058153,3.9972265090214774);
\draw [line width=0.8pt,color=zzttqq] (-3.3369438781143117,-0.7118007455449216)-- (-0.7450035865221785,-2.147562489360967);
\draw [line width=0.8pt,color=zzttqq] (-0.7450035865221785,-2.147562489360967)-- (-0.7450035865221785,-3.1477377807076374);
\draw [line width=0.8pt,color=zzttqq] (-0.7450035865221785,-3.1477377807076374)-- (-2.2874464123395972,-4.455920507443929);
\draw [line width=0.8pt,color=zzttqq] (-2.2874464123395972,-4.455920507443929)-- (-3.153471816124036,-4.955920507443929);
\draw [line width=0.8pt,color=zzttqq] (-3.153471816124036,-4.955920507443929)-- (-5.0039026970254294,-4.902660907542389);
\draw [line width=0.8pt,color=zzttqq] (-5.0039026970254294,-4.902660907542389)-- (-5.869928100809869,-4.402660907542388);
\draw [line width=0.8pt,color=zzttqq] (-5.869928100809869,-4.402660907542388)-- (-5.869928100809869,-3.402660907542388);
\draw [line width=0.8pt,color=zzttqq] (-5.869928100809869,-3.402660907542388)-- (-5.751548027477352,-2.4559205074439285);
\draw [line width=0.8pt,color=zzttqq] (-5.751548027477352,-2.4559205074439285)-- (-4.202969281898751,-1.211800745544921);
\draw [line width=0.8pt,color=zzttqq] (-4.202969281898751,-1.211800745544921)-- (-3.3369438781143117,-0.7118007455449216);
\draw [line width=0.8pt,color=zzttqq] (-1.6087380393548398,-0.7195401930625681)-- (0.9689274855951311,-2.176946921938492);
\draw [line width=0.8pt,color=zzttqq] (0.9689274855951311,-2.176946921938492)-- (0.9689274855951311,-3.177122213285162);
\draw [line width=0.8pt,color=zzttqq] (0.9689274855951311,-3.177122213285162)-- (-0.5735153402222877,-4.485304940021454);
\draw [line width=0.8pt,color=zzttqq] (-0.5735153402222877,-4.485304940021454)-- (-1.4395407440067265,-4.985304940021454);
\draw [line width=0.8pt,color=zzttqq] (-1.4395407440067265,-4.985304940021454)-- (-3.153471816124036,-4.955920507443929);
\draw [line width=0.8pt,color=zzttqq] (-3.153471816124036,-4.955920507443929)-- (-4.165156253417765,-4.344050948617536);
\draw [line width=0.8pt,color=zzttqq] (-4.165156253417765,-4.344050948617536)-- (-4.147646539622813,-3.3459972623052985);
\draw [line width=0.8pt,color=zzttqq] (-4.147646539622813,-3.3459972623052985)-- (-4.025078543058153,-2.4530018587627698);
\draw [line width=0.8pt,color=zzttqq] (-4.025078543058153,-2.4530018587627698)-- (-2.423946221123682,-1.2175404465530901);
\draw [line width=0.8pt,color=zzttqq] (-2.423946221123682,-1.2175404465530901)-- (-1.6087380393548398,-0.7195401930625681);
\draw [line width=0.8pt,color=zzttqq] (-2.490961998428125,0.7847345160111372)-- (0.10721391255041136,-0.7370499068575196);
\draw [line width=0.8pt,color=zzttqq] (0.10721391255041136,-0.7370499068575196)-- (0.08914206919101209,-1.6795180765652522);
\draw [line width=0.8pt,color=zzttqq] (0.08914206919101209,-1.6795180765652522)-- (-1.4686603289952274,-2.9257641312264613);
\draw [line width=0.8pt,color=zzttqq] (-1.4686603289952274,-2.9257641312264613)-- (-2.326636304947853,-3.4160361174851044);
\draw [line width=0.8pt,color=zzttqq] (-2.326636304947853,-3.4160361174851044)-- (-4.147646539622812,-3.345997262305295);
\draw [line width=0.8pt,color=zzttqq] (-4.147646539622812,-3.345997262305295)-- (-4.988112801780487,-2.873234989841607);
\draw [line width=0.8pt,color=zzttqq] (-4.988112801780487,-2.873234989841607)-- (-5.059328979528317,-1.8997954801728507);
\draw [line width=0.8pt,color=zzttqq] (-5.059328979528317,-1.8997954801728507)-- (-4.905566147791165,-0.9593852458878699);
\draw [line width=0.8pt,color=zzttqq] (-4.905566147791165,-0.9593852458878699)-- (-3.356987402212564,0.28473451601113786);
\draw [line width=0.8pt,color=zzttqq] (-3.356987402212564,0.28473451601113786)-- (-2.490961998428125,0.7847345160111372);
\draw [line width=0.8pt,color=zzttqq] (0.17412999022963405,-2.906125645986329)-- (1.0401553940140733,-3.4061256459863296);
\draw [line width=0.8pt,color=zzttqq] (1.0401553940140733,-3.4061256459863296)-- (2.8905862749154667,-3.4593852458878707);
\draw [line width=0.8pt,color=zzttqq] (2.8905862749154667,-3.4593852458878707)-- (3.7566116786999055,-2.9593852458878707);
\draw [line width=0.8pt,color=zzttqq] (3.7566116786999055,-2.9593852458878707)-- (5.299054504517324,-1.6512025191515785);
\draw [line width=0.8pt,color=zzttqq] (5.299054504517324,-1.6512025191515785)-- (5.299054504517324,-0.6510272278049084);
\draw [line width=0.8pt,color=zzttqq] (5.299054504517324,-0.6510272278049084)-- (2.707114212925191,0.7847345160111372);
\draw [line width=0.8pt,color=zzttqq] (2.707114212925191,0.7847345160111372)-- (1.875695005840517,0.3485523484294755);
\draw [line width=0.8pt,color=zzttqq] (1.875695005840517,0.3485523484294755)-- (0.2823110504999268,-0.9646761861918896);
\draw [line width=0.8pt,color=zzttqq] (0.2823110504999268,-0.9646761861918896)-- (0.17725276773021753,-1.77012302075966);
\draw [line width=0.8pt,color=zzttqq] (0.17725276773021753,-1.77012302075966)-- (0.17412999022963405,-2.906125645986329);
\draw [line width=0.8pt,color=zzttqq] (-2.489038209781441,0.7847345160111372)-- (-3.35506361356588,0.28473451601113786);
\draw [line width=0.8pt,color=zzttqq] (-3.35506361356588,0.28473451601113786)-- (-4.903642359144481,-0.9593852458878699);
\draw [line width=0.8pt,color=zzttqq] (-4.903642359144481,-0.9593852458878699)-- (-5.022022432476998,-1.9061256459863292);
\draw [line width=0.8pt,color=zzttqq] (-5.022022432476998,-1.9061256459863292)-- (-5.022022432476998,-2.906125645986329);
\draw [line width=0.8pt,color=zzttqq] (-5.022022432476998,-2.906125645986329)-- (-4.11262711203291,-3.381016689895201);
\draw [line width=0.8pt,color=zzttqq] (-4.11262711203291,-3.381016689895201)-- (-2.3055661477911653,-3.4593852458878707);
\draw [line width=0.8pt,color=zzttqq] (-2.3055661477911653,-3.4593852458878707)-- (-1.4511506152002758,-3.0259196941335835);
\draw [line width=0.8pt,color=zzttqq] (-1.4511506152002758,-3.0259196941335835)-- (0.,-1.595025882810145);
\draw [line width=0.8pt,color=zzttqq] (0.,-1.595025882810145)-- (0.10721391255041136,-0.7370499068575196);
\draw [line width=0.8pt,color=zzttqq] (0.10721391255041136,-0.7370499068575196)-- (-2.489038209781441,0.7847345160111372);
\draw [line width=0.8pt,color=zzttqq] (0.07219448496050827,0.3835717760193786)-- (0.08970419875545982,-0.649501337882762);
\draw [line width=0.8pt,color=zzttqq] (0.08970419875545982,-0.649501337882762)-- (-1.4395407440067265,-1.9593852458878707);
\draw [line width=0.8pt,color=zzttqq] (-1.4395407440067265,-1.9593852458878707)-- (-2.3055661477911653,-2.4593852458878707);
\draw [line width=0.8pt,color=zzttqq] (-2.3055661477911653,-2.4593852458878707)-- (-4.140556581794558,-2.3905514570939643);
\draw [line width=0.8pt,color=zzttqq] (-4.140556581794558,-2.3905514570939643)-- (-5.022022432476998,-1.9061256459863292);
\draw [line width=0.8pt,color=zzttqq] (-5.022022432476998,-1.9061256459863292)-- (-5.022022432476998,-0.9061256459863293);
\draw [line width=0.8pt,color=zzttqq] (-5.022022432476998,-0.9061256459863293)-- (-4.970603087985536,0.);
\draw [line width=0.8pt,color=zzttqq] (-4.970603087985536,0.)-- (-3.35506361356588,1.2847345160111379);
\draw [line width=0.8pt,color=zzttqq] (-3.35506361356588,1.2847345160111379)-- (-2.489038209781441,1.7847345160111372);
\draw [line width=0.8pt,color=zzttqq] (-2.489038209781441,1.7847345160111372)-- (0.07219448496050827,0.3835717760193786);
\draw [line width=0.8pt,color=zzttqq] (-0.7569874022125642,1.7847345160111379)-- (1.834952889379569,0.3489727721950923);
\draw [line width=0.8pt,color=zzttqq] (1.834952889379569,0.3489727721950923)-- (1.834952889379569,-0.6512025191515778);
\draw [line width=0.8pt,color=zzttqq] (1.834952889379569,-0.6512025191515778)-- (0.29251006356215037,-1.95938524588787);
\draw [line width=0.8pt,color=zzttqq] (0.29251006356215037,-1.95938524588787)-- (-0.5735153402222883,-2.4593852458878698);
\draw [line width=0.8pt,color=zzttqq] (-0.5735153402222883,-2.4593852458878698)-- (-2.361655732537756,-2.3654532897880123);
\draw [line width=0.8pt,color=zzttqq] (-2.361655732537756,-2.3654532897880123)-- (-3.289670563670188,-1.8401618759394662);
\draw [line width=0.8pt,color=zzttqq] (-3.289670563670188,-1.8401618759394662)-- (-3.307180277465139,-0.8596179034221804);
\draw [line width=0.8pt,color=zzttqq] (-3.307180277465139,-0.8596179034221804)-- (-3.2196317084903816,0.);
\draw [line width=0.8pt,color=zzttqq] (-3.2196317084903816,0.)-- (-1.625765739537775,1.31246769611255);
\draw [line width=0.8pt,color=zzttqq] (-1.625765739537775,1.31246769611255)-- (-0.7569874022125642,1.7847345160111379);
\draw [line width=0.8pt,color=zzttqq] (0.057239039181291695,1.3337074691501623)-- (0.9750634053563134,1.7847345160111372);
\draw [line width=0.8pt,color=zzttqq] (0.9750634053563134,1.7847345160111372)-- (3.567003696948447,0.3489727721950916);
\draw [line width=0.8pt,color=zzttqq] (3.567003696948447,0.3489727721950916)-- (3.567003696948447,-0.6512025191515786);
\draw [line width=0.8pt,color=zzttqq] (3.567003696948447,-0.6512025191515786)-- (2.024560871131028,-1.9593852458878707);
\draw [line width=0.8pt,color=zzttqq] (2.024560871131028,-1.9593852458878707)-- (1.1585354673465895,-2.4593852458878707);
\draw [line width=0.8pt,color=zzttqq] (1.1585354673465895,-2.4593852458878707)-- (-0.8630928130169472,-2.444173588806049);
\draw [line width=0.8pt,color=zzttqq] (-0.8630928130169472,-2.444173588806049)-- (-1.5386991841750335,-1.8576715897344178);
\draw [line width=0.8pt,color=zzttqq] (-1.5386991841750335,-1.8576715897344178)-- (-1.5912283255598882,-0.7895790482423742);
\draw [line width=0.8pt,color=zzttqq] (-1.5912283255598882,-0.7895790482423742)-- (-1.4511506152002758,0.15594549668500865);
\draw [line width=0.8pt,color=zzttqq] (-1.4511506152002758,0.15594549668500865)-- (0.057239039181291695,1.3337074691501623);
\draw [line width=0.8pt,color=zzttqq] (1.8320101678787537,1.3100432655576164)-- (2.707114212925191,1.7847345160111372);
\draw [line width=0.8pt,color=zzttqq] (2.707114212925191,1.7847345160111372)-- (5.299054504517324,0.3489727721950916);
\draw [line width=0.8pt,color=zzttqq] (5.299054504517324,0.3489727721950916)-- (5.299054504517324,-0.6512025191515786);
\draw [line width=0.8pt,color=zzttqq] (5.299054504517324,-0.6512025191515786)-- (3.7566116786999055,-1.9593852458878707);
\draw [line width=0.8pt,color=zzttqq] (3.7566116786999055,-1.9593852458878707)-- (2.8905862749154667,-2.4593852458878707);
\draw [line width=0.8pt,color=zzttqq] (2.8905862749154667,-2.4593852458878707)-- (1.1585354673465895,-2.4593852458878707);
\draw [line width=0.8pt,color=zzttqq] (1.1585354673465895,-2.4593852458878707)-- (0.19476248152516906,-1.8576715897344178);
\draw [line width=0.8pt,color=zzttqq] (0.19476248152516906,-1.8576715897344178)-- (0.1937268067341193,-0.8881517759828096);
\draw [line width=0.8pt,color=zzttqq] (0.1937268067341193,-0.8881517759828096)-- (0.2998207642948783,0.);
\draw [line width=0.8pt,color=zzttqq] (0.2998207642948783,0.)-- (1.8320101678787537,1.3100432655576164);
\draw [line width=0.8pt,color=zzttqq] (0.9750634053563132,1.7847345160111372)-- (1.841088809140752,2.284734516011137);
\draw [line width=0.8pt,color=zzttqq] (1.841088809140752,2.284734516011137)-- (4.4191667564150965,0.8363717958823323);
\draw [line width=0.8pt,color=zzttqq] (4.4191667564150965,0.8363717958823323)-- (4.414603506108491,-0.159229351624119);
\draw [line width=0.8pt,color=zzttqq] (4.414603506108491,-0.159229351624119)-- (2.873748692152755,-1.437438458655581);
\draw [line width=0.8pt,color=zzttqq] (2.873748692152755,-1.437438458655581)-- (2.0245608711310306,-1.9593852458878702);
\draw [line width=0.8pt,color=zzttqq] (2.0245608711310306,-1.9593852458878702)-- (-0.6281940668375534,-1.4199287448606295);
\draw [line width=0.8pt,color=zzttqq] (-0.6281940668375534,-1.4199287448606295)-- (-0.6281940668375534,-0.36934591716353743);
\draw [line width=0.8pt,color=zzttqq] (-0.6281940668375534,-0.36934591716353743)-- (-0.5581552116577472,0.5236494863789909);
\draw [line width=0.8pt,color=zzttqq] (-0.5581552116577472,0.5236494863789909)-- (0.9750634053563132,1.7847345160111372);
\draw [line width=0.8pt,color=zzttqq] (1.8410888091407522,3.284734516011138)-- (2.707114212925191,3.784734516011137);
\draw [line width=0.8pt,color=zzttqq] (2.707114212925191,3.784734516011137)-- (5.299054504517324,2.3489727721950917);
\draw [line width=0.8pt,color=zzttqq] (5.299054504517324,2.3489727721950917)-- (5.299054504517324,1.3487974808484215);
\draw [line width=0.8pt,color=zzttqq] (5.299054504517324,1.3487974808484215)-- (3.6967052405154774,0.);
\draw [line width=0.8pt,color=zzttqq] (3.6967052405154774,0.)-- (2.8905862749154667,-0.45938524588787066);
\draw [line width=0.8pt,color=zzttqq] (2.8905862749154667,-0.45938524588787066)-- (0.8660254037844386,-0.5);
\draw [line width=0.8pt,color=zzttqq] (0.8660254037844386,-0.5)-- (0.26480133670497524,0.17345521047996018);
\draw [line width=0.8pt,color=zzttqq] (0.26480133670497524,0.17345521047996018)-- (0.1741299902296345,1.0938743540136708);
\draw [line width=0.8pt,color=zzttqq] (0.1741299902296345,1.0938743540136708)-- (0.2925100635621507,2.0406147541121302);
\draw [line width=0.8pt,color=zzttqq] (0.2925100635621507,2.0406147541121302)-- (1.8410888091407522,3.284734516011138);
\draw [line width=0.8pt,color=zzttqq] (0.9750634053563134,3.784734516011137)-- (3.5391178163609136,2.3796791486438535);
\draw [line width=0.8pt,color=zzttqq] (3.5391178163609136,2.3796791486438535)-- (3.521608102565962,1.3290963209467614);
\draw [line width=0.8pt,color=zzttqq] (3.521608102565962,1.3290963209467614)-- (1.9632435748152748,0.);
\draw [line width=0.8pt,color=zzttqq] (1.9632435748152748,0.)-- (1.1056846271618606,-0.4951921917651505);
\draw [line width=0.8pt,color=zzttqq] (1.1056846271618606,-0.4951921917651505)-- (-0.6281940668375534,-0.38685563095848896);
\draw [line width=0.8pt,color=zzttqq] (-0.6281940668375534,-0.38685563095848896)-- (-1.5579208173392434,0.09387435401367084);
\draw [line width=0.8pt,color=zzttqq] (-1.5579208173392434,0.09387435401367084)-- (-1.557920817339243,1.0938743540136708);
\draw [line width=0.8pt,color=zzttqq] (-1.557920817339243,1.0938743540136708)-- (-1.4395407440067267,2.0406147541121302);
\draw [line width=0.8pt,color=zzttqq] (-1.4395407440067267,2.0406147541121302)-- (0.10903800157187482,3.284734516011138);
\draw [line width=0.8pt,color=zzttqq] (0.10903800157187482,3.284734516011138)-- (0.9750634053563134,3.784734516011137);
\draw [line width=0.8pt,color=zzttqq] (-0.6281940668375534,-0.38685563095848896)-- (0.31733047808982984,0.);
\draw [line width=0.8pt,color=zzttqq] (0.31733047808982984,0.)-- (1.8349528893795697,1.3487974808484215);
\draw [line width=0.8pt,color=zzttqq] (1.8349528893795697,1.3487974808484215)-- (1.8349528893795697,2.3489727721950917);
\draw [line width=0.8pt,color=zzttqq] (1.8349528893795697,2.3489727721950917)-- (-0.7569874022125638,3.784734516011137);
\draw [line width=0.8pt,color=zzttqq] (-0.7569874022125638,3.784734516011137)-- (-1.6230128059970024,3.284734516011138);
\draw [line width=0.8pt,color=zzttqq] (-1.6230128059970024,3.284734516011138)-- (-3.171591551575604,2.0406147541121302);
\draw [line width=0.8pt,color=zzttqq] (-3.171591551575604,2.0406147541121302)-- (-3.2899716249081203,1.0938743540136708);
\draw [line width=0.8pt,color=zzttqq] (-3.2899716249081203,1.0938743540136708)-- (-3.2899716249081203,0.09387435401367084);
\draw [line width=0.8pt,color=zzttqq] (-3.2899716249081203,0.09387435401367084)-- (-2.423946221123681,-0.4061256459863296);
\draw [line width=0.8pt,color=zzttqq] (-2.423946221123681,-0.4061256459863296)-- (-0.6281940668375534,-0.38685563095848896);
\draw [line width=0.8pt,color=zzttqq] (-2.489038209781441,3.784734516011137)-- (0.08970419875545982,2.309640293464047);
\draw [line width=0.8pt,color=zzttqq] (0.08970419875545982,2.309640293464047)-- (0.10721391255041136,1.346606034741713);
\draw [line width=0.8pt,color=zzttqq] (0.10721391255041136,1.346606034741713)-- (-1.4336409014053244,0.);
\draw [line width=0.8pt,color=zzttqq] (-1.4336409014053244,0.)-- (-2.2741055284164906,-0.40451772231303323);
\draw [line width=0.8pt,color=zzttqq] (-2.2741055284164906,-0.40451772231303323)-- (-4.155997028692559,-0.4061256459863296);
\draw [line width=0.8pt,color=zzttqq] (-4.155997028692559,-0.4061256459863296)-- (-5.022022432476998,0.09387435401367084);
\draw [line width=0.8pt,color=zzttqq] (-5.022022432476998,0.09387435401367084)-- (-5.022022432476998,1.0938743540136708);
\draw [line width=0.8pt,color=zzttqq] (-5.022022432476998,1.0938743540136708)-- (-4.903642359144481,2.0406147541121302);
\draw [line width=0.8pt,color=zzttqq] (-4.903642359144481,2.0406147541121302)-- (-3.35506361356588,3.284734516011138);
\draw [line width=0.8pt,color=zzttqq] (-3.35506361356588,3.284734516011138)-- (-2.489038209781441,3.784734516011137);
\draw [line width=0.8pt,color=zzttqq] (-3.3550636135658793,5.284734516011137)-- (-0.7631233219737461,3.8489727721950917);
\draw [line width=0.8pt,color=zzttqq] (-0.7631233219737461,3.8489727721950917)-- (-0.7631233219737461,2.8487974808484213);
\draw [line width=0.8pt,color=zzttqq] (-0.7631233219737461,2.8487974808484213)-- (-2.305566147791165,1.5406147541121293);
\draw [line width=0.8pt,color=zzttqq] (-2.305566147791165,1.5406147541121293)-- (-3.1715915515756037,1.0406147541121293);
\draw [line width=0.8pt,color=zzttqq] (-3.1715915515756037,1.0406147541121293)-- (-5.022022432476997,1.0938743540136704);
\draw [line width=0.8pt,color=zzttqq] (-5.022022432476997,1.0938743540136704)-- (-5.888047836261436,1.5938743540136708);
\draw [line width=0.8pt,color=zzttqq] (-5.888047836261436,1.5938743540136708)-- (-5.8880478362614355,2.593874354013671);
\draw [line width=0.8pt,color=zzttqq] (-5.8880478362614355,2.593874354013671)-- (-5.76966776292892,3.5406147541121302);
\draw [line width=0.8pt,color=zzttqq] (-5.76966776292892,3.5406147541121302)-- (-4.221089017350318,4.784734516011138);
\draw [line width=0.8pt,color=zzttqq] (-4.221089017350318,4.784734516011138)-- (-3.3550636135658793,5.284734516011137);
\draw [line width=0.8pt,color=zzttqq] (-4.025078543058153,3.605359114290461)-- (-4.147646539622813,2.659834569363078);
\draw [line width=0.8pt,color=zzttqq] (-4.147646539622813,2.659834569363078)-- (-4.147646539622813,1.6617808830508407);
\draw [line width=0.8pt,color=zzttqq] (-4.147646539622813,1.6617808830508407)-- (-3.237141422285333,1.118979755407343);
\draw [line width=0.8pt,color=zzttqq] (-3.237141422285333,1.118979755407343)-- (-1.4395407440067265,1.0406147541121293);
\draw [line width=0.8pt,color=zzttqq] (-1.4395407440067265,1.0406147541121293)-- (-0.5735153402222877,1.5406147541121293);
\draw [line width=0.8pt,color=zzttqq] (-0.5735153402222877,1.5406147541121293)-- (0.9689274855951311,2.8487974808484213);
\draw [line width=0.8pt,color=zzttqq] (0.9689274855951311,2.8487974808484213)-- (0.9689274855951311,3.8489727721950917);
\draw [line width=0.8pt,color=zzttqq] (0.9689274855951311,3.8489727721950917)-- (-1.6230128059970024,5.284734516011137);
\draw [line width=0.8pt,color=zzttqq] (-1.6230128059970024,5.284734516011137)-- (-2.508338561978974,4.815705410463189);
\draw [line width=0.8pt,color=zzttqq] (-2.508338561978974,4.815705410463189)-- (-4.025078543058153,3.605359114290461);
\draw [line width=0.8pt,color=zzttqq] (2.7009782931640083,3.8489727721950917)-- (2.7009782931640083,2.8487974808484213);
\draw [line width=0.8pt,color=zzttqq] (2.7009782931640083,2.8487974808484213)-- (1.1585354673465895,1.5406147541121293);
\draw [line width=0.8pt,color=zzttqq] (1.1585354673465895,1.5406147541121293)-- (0.2925100635621508,1.0406147541121293);
\draw [line width=0.8pt,color=zzttqq] (0.2925100635621508,1.0406147541121293)-- (-1.5386991841750335,1.153999182997246);
\draw [line width=0.8pt,color=zzttqq] (-1.5386991841750335,1.153999182997246)-- (-2.4316945877175624,1.6617808830508407);
\draw [line width=0.8pt,color=zzttqq] (-2.4316945877175624,1.6617808830508407)-- (-2.466714015307465,2.6773442831580296);
\draw [line width=0.8pt,color=zzttqq] (-2.466714015307465,2.6773442831580296)-- (-2.326636304947853,3.570339686700558);
\draw [line width=0.8pt,color=zzttqq] (-2.326636304947853,3.570339686700558)-- (-0.7752081608871471,4.815107387783365);
\draw [line width=0.8pt,color=zzttqq] (-0.7752081608871471,4.815107387783365)-- (0.10903800157187488,5.284734516011137);
\draw [line width=0.8pt,color=zzttqq] (0.10903800157187488,5.284734516011137)-- (2.7009782931640083,3.8489727721950917);
\draw [line width=0.8pt,color=zzttqq] (0.9832739804525663,4.800466194855695)-- (1.841088809140752,5.284734516011137);
\draw [line width=0.8pt,color=zzttqq] (1.841088809140752,5.284734516011137)-- (4.4330291007328855,3.8489727721950917);
\draw [line width=0.8pt,color=zzttqq] (4.4330291007328855,3.8489727721950917)-- (4.4330291007328855,2.8487974808484213);
\draw [line width=0.8pt,color=zzttqq] (4.4330291007328855,2.8487974808484213)-- (2.8905862749154667,1.5406147541121293);
\draw [line width=0.8pt,color=zzttqq] (2.8905862749154667,1.5406147541121293)-- (2.024560871131028,1.0406147541121293);
\draw [line width=0.8pt,color=zzttqq] (2.024560871131028,1.0406147541121293)-- (0.14223334014031444,1.1364894692022947);
\draw [line width=0.8pt,color=zzttqq] (0.14223334014031444,1.1364894692022947)-- (-0.6457037806325049,1.609251741665986);
\draw [line width=0.8pt,color=zzttqq] (-0.6457037806325049,1.609251741665986)-- (-0.6982329220173595,2.659834569363078);
\draw [line width=0.8pt,color=zzttqq] (-0.6982329220173595,2.659834569363078)-- (-0.6106843530426018,3.6403785418803642);
\draw [line width=0.8pt,color=zzttqq] (-0.6106843530426018,3.6403785418803642)-- (0.9832739804525663,4.800466194855695);
\draw [line width=0.8pt,color=zzttqq] (3.5731396167096294,5.284734516011137)-- (6.165079908301763,3.8489727721950917);
\draw [line width=0.8pt,color=zzttqq] (6.165079908301763,3.8489727721950917)-- (6.165079908301763,2.8487974808484213);
\draw [line width=0.8pt,color=zzttqq] (6.165079908301763,2.8487974808484213)-- (4.622637082484344,1.5406147541121293);
\draw [line width=0.8pt,color=zzttqq] (4.622637082484344,1.5406147541121293)-- (3.756611678699905,1.0406147541121293);
\draw [line width=0.8pt,color=zzttqq] (3.756611678699905,1.0406147541121293)-- (1.7320508075688772,1.);
\draw [line width=0.8pt,color=zzttqq] (1.7320508075688772,1.)-- (0.9826996022979885,1.679290596845792);
\draw [line width=0.8pt,color=zzttqq] (0.9826996022979885,1.679290596845792)-- (1.0352287436828431,2.6773442831580296);
\draw [line width=0.8pt,color=zzttqq] (1.0352287436828431,2.6773442831580296)-- (1.1577967402475038,3.6228688280854127);
\draw [line width=0.8pt,color=zzttqq] (1.1577967402475038,3.6228688280854127)-- (2.7535137327049597,4.779312071054633);
\draw [line width=0.8pt,color=zzttqq] (2.7535137327049597,4.779312071054633)-- (3.5731396167096294,5.284734516011137);
\draw [line width=0.8pt,color=zzttqq] (1.8349528893795688,-0.6512025191515778)-- (4.484642361288297,-2.1728464380435457);
\draw [line width=0.8pt,color=zzttqq] (4.484642361288297,-2.1728464380435457)-- (4.4330291007328855,-3.1512025191515787);
\draw [line width=0.8pt,color=zzttqq] (4.4330291007328855,-3.1512025191515787)-- (2.8905862749154667,-4.459385245887871);
\draw [line width=0.8pt,color=zzttqq] (2.8905862749154667,-4.459385245887871)-- (2.024560871131028,-4.959385245887871);
\draw [line width=0.8pt,color=zzttqq] (2.024560871131028,-4.959385245887871)-- (0.1741299902296345,-4.90612564598633);
\draw [line width=0.8pt,color=zzttqq] (0.1741299902296345,-4.90612564598633)-- (-0.6450178235220853,-4.414089803797341);
\draw [line width=0.8pt,color=zzttqq] (-0.6450178235220853,-4.414089803797341)-- (-0.7149822736075474,-3.345997262305298);
\draw [line width=0.8pt,color=zzttqq] (-0.7149822736075474,-3.345997262305298)-- (-0.5153403947258817,-2.4257979254454276);
\draw [line width=0.8pt,color=zzttqq] (-0.5153403947258817,-2.4257979254454276)-- (0.9466734159221789,-1.227321893116162);
\draw [line width=0.8pt,color=zzttqq] (0.9466734159221789,-1.227321893116162)-- (1.8349528893795688,-0.6512025191515778);
\draw [line width=0.8pt,color=zzttqq] (3.5670036969484467,-0.6512025191515785)-- (6.165079908301763,-2.1510272278049083);
\draw [line width=0.8pt,color=zzttqq] (6.165079908301763,-2.1510272278049083)-- (6.165079908301763,-3.1512025191515787);
\draw [line width=0.8pt,color=zzttqq] (6.165079908301763,-3.1512025191515787)-- (4.622637082484344,-4.459385245887871);
\draw [line width=0.8pt,color=zzttqq] (4.622637082484344,-4.459385245887871)-- (3.756611678699905,-4.959385245887871);
\draw [line width=0.8pt,color=zzttqq] (3.756611678699905,-4.959385245887871)-- (1.7334609450754725,-4.951006729428215);
\draw [line width=0.8pt,color=zzttqq] (1.7334609450754725,-4.951006729428215)-- (1.1040934286144688,-4.39658009000239);
\draw [line width=0.8pt,color=zzttqq] (1.1040934286144688,-4.39658009000239)-- (1.0401553940140729,-3.406125645986329);
\draw [line width=0.8pt,color=zzttqq] (1.0401553940140729,-3.406125645986329)-- (1.1585354673465897,-2.4593852458878707);
\draw [line width=0.8pt,color=zzttqq] (1.1585354673465897,-2.4593852458878707)-- (2.721564543722245,-1.1603592511204144);
\draw [line width=0.8pt,color=zzttqq] (2.721564543722245,-1.1603592511204144)-- (3.5670036969484467,-0.6512025191515785);
\begin{scriptsize}
\draw [fill=qqqqff] (18.527432824839444,3.325203693571236) circle (1.5pt);
\draw [fill=uuuuuu] (0.36777928298570156,0.058105479348663856) circle (2.0pt);
\draw [fill=uuuuuu] (0.8660254037844386,0.5) circle (2.0pt);
\draw [fill=uuuuuu] (0.,1.) circle (2.0pt);
\draw [fill=uuuuuu] (0.10721391255003894,-0.7370499068574413) circle (2.0pt);
\draw [fill=uuuuuu] (0.10721391255003894,-0.7370499068574413) circle (2.0pt);
\draw [fill=uququq] (-0.6249628864424047,-0.32814836712539197) circle (2.5pt);
\draw [fill=uuuuuu] (0.8660254037844386,-0.4835929279903689) circle (2.0pt);
\draw [fill=uuuuuu] (-0.5581552116577472,0.5236494863789908) circle (2.0pt);
\draw [fill=uuuuuu] (0.07219448496050827,0.3835717760193786) circle (2.0pt);
\draw [fill=uuuuuu] (-0.1908506870931594,-0.8898123044317393) circle (2.0pt);
\draw [fill=uuuuuu] (-0.8660254037844392,-0.38940777236960694) circle (2.0pt);
\draw [fill=uququq] (0.08567425463669892,-0.41173463487587025) circle (2.5pt);
\draw [fill=uququq] (0.1993905879184362,0.8374500414421162) circle (2.5pt);
\draw [fill=uuuuuu] (-0.8660254037844389,0.5) circle (2.0pt);
\draw [fill=black] (-0.39235075264527347,0.5739501187353968) circle (2.5pt);
\draw [fill=black] (-0.8541161232093236,-0.13869523271913) circle (2.5pt);
\draw [fill=uququq] (0.0792603638956457,-0.033315081153729685) circle (2.5pt);
\draw [fill=uququq] (0.2731074142752643,-0.44317093320635537) circle (2.5pt);
\draw [fill=uququq] (0.26111141653153824,0.2109121723291624) circle (2.5pt);
\draw [fill=uququq] (0.2093083155308493,-0.8791557876846371) circle (2.5pt);
\draw [fill=uququq] (-0.6281940668375535,-0.6515083990629744) circle (2.5pt);
\end{scriptsize}
\end{tikzpicture}